\documentclass[aps,prla,twocolumn,nofootinbib,superscriptaddress]{revtex4-2}
\usepackage{amsmath}
\usepackage{amssymb}
\usepackage{mathtools}
\usepackage{graphicx}
\usepackage{braket}
\usepackage{hyperref}
\hypersetup{colorlinks=true,linkcolor=blue,citecolor=red,filecolor=red,urlcolor=blue,runcolor=blue}
\usepackage{paralist}
\usepackage{color}
\usepackage{mathrsfs}
\usepackage{amsthm}
\usepackage{centernot}
\usepackage{stmaryrd}
\usepackage{xcolor}
\usepackage{soul}
\usepackage{thmtools} 
\usepackage{thm-restate}
\usepackage{upgreek}

\newtheorem*{proposition*}{Proposition}
\newtheorem*{theorem*}{Theorem}
\newtheorem{lemma}{Lemma}
\newcommand{\tr}[0]{\textnormal{Tr}}

\DeclareMathOperator*{\Tr}{Tr}
\DeclareMathOperator*{\Rank}{Rank}
\DeclareMathOperator*{\Ran}{Ran}
\DeclareMathOperator*{\Span}{Span}

\DeclareMathOperator*{\Real}{Re}

\DeclareMathOperator{\Prob}{Prob}

\DeclareRobustCommand{\openzero}{\leavevmode\hbox{0\kern-.55em0}}

\newcommand{\prlsection}[1]{{\em {#1}.---~}}

\begin{document}

\title{Information Scrambling over Bipartitions:\\Equilibration, Entropy Production, and Typicality}

\date{\today}

\author{Georgios Styliaris}
%\email [e-mail address: ]{styliari@usc.edu}

\affiliation{Max-Planck-Institut f\"ur Quantenoptik, Hans-Kopfermann-Str. 1, 85748 Garching, Germany}

\affiliation{Munich Center for Quantum Science and Technology, Schellingstraße 4, 80799 M\"unchen, Germany}

\affiliation{Department of Physics and Astronomy, and Center for Quantum Information
Science and Technology, University of Southern California, Los Angeles, California 90089, USA}

\author{Namit Anand}

\affiliation{Department of Physics and Astronomy, and Center for Quantum Information
Science and Technology, University of Southern California, Los Angeles, California 90089, USA}

\author{Paolo Zanardi}

\affiliation{Department of Physics and Astronomy, and Center for Quantum Information
Science and Technology, University of Southern California, Los Angeles, California 90089, USA}

\begin{abstract}

In recent years, the out-of-time-order correlator (OTOC) has emerged as a diagnostic tool for information scrambling in quantum many-body systems.  Here, we present exact analytical results for the OTOC for a typical pair of random local operators supported over two regions of a bipartition. Quite remarkably, we show that this ``bipartite  OTOC'' is equal to the operator entanglement of the evolution and we determine its interplay with entangling power. Furthermore, we compute long-time averages of the OTOC and reveal their connection with eigenstate entanglement. For Hamiltonian systems, we uncover a hierarchy of constraints over the structure of the spectrum and elucidate how this affects the equilibration value of the OTOC. Finally, we provide operational significance to this bipartite  OTOC by unraveling intimate connections with average entropy production and scrambling of information at the level of quantum channels.

%In recent years, the out-of-time-ordered correlator (OTOC) has emerged as a diagnostic tool for information scrambling in quantum many-body systems. Here, we provide exact analytical results for the long-time averages of the OTOC for a typical pair of random local operators supported over two regions of a bipartition, thereby revealing its connection with eigenstate entanglement. We uncover a hierarchy of constraints over the structure of the spectrum of Hamiltonian systems, for instance integrable models, and elucidate how they affect the equilibration value of the OTOC. We provide operational significance to this ``bipartite OTOC,'' by unraveling intimate connections with operator entanglement, average entropy production, and scrambling of information at the level of quantum channels.

\end{abstract}

\maketitle

%Interacting quantum many-body systems are information spreaders. Encoding some property into the well-defined value of a local observable will generically lead, under the time-evolution, to its leakage over nonlocal degrees of freedom. This hallmark of certain quantum evolutions is often referred to as scrambling~\cite{page1993average, hayden2007black, hosur2016chaos, von2018operator, moudgalya2019operator}; although information encoded in the initial preparation lies within the closed system, it is largely inaccessible by local measurements. %Scrambling is in close connection also with the notion operator spreading~\cite{}. 

\prlsection{Introduction} A characteristic feature of certain quantum many-body systems is their ability to quickly spread ``localized'' information over subsystems, thereby making it inaccessible to local observables. Although unitary evolution retains all information, this local inaccessibility manifests itself as equilibration in closed systems,
%(even in closed quantum systems)
and has been termed ``information scrambling''~\cite{page1993average, hayden2007black, hosur2016chaos, von2018operator, moudgalya2019operator}.

%\red{A characteristic feature of chaotic quantum systems is their ability to quickly spread ``localized'' information over subsystems, thereby rendering it inaccessible to local observables. Although unitary evolution retains all information, such local inaccessibility gives the appearance of thermalization (even in closed quantum systems), a phenomena that has been termed ``information scrambling''.}

For Hamiltonian quantum dynamics, scrambling can be probed by examining the overlap of a time-evolved local operator $V(t) \coloneqq U_t^\dagger V U_t$ with a second static operator $W$. This overlap is commonly quantified via the strength of the commutator\footnote{In fact,  $C_{V,W}(t) = \dfrac{1}{2} \big\|  \big[ V(t) , W \big] \big\|^2$ for the norm associated with the inner product $\braket{X,Y}_{\beta} = \tr \big( X^\dagger Y \rho_\beta \big)$,  $\beta< \infty$.}
\begin{align} \label{eq:def_thermal_OTOC}
C_{V,W}(t)  \coloneqq  \frac{1}{2}  \Tr \big(  \left[ V(t) , W \right]^\dagger  \left[ V(t) , W \right] \rho_\beta \big)
\end{align} 
where $\rho_\beta$ denotes the thermal state at inverse temperature $\beta$. From the perspective of information spreading, $C_{V,W}(t)$ is a natural quantity to consider since it constitutes a state-dependent variant of the Lieb-Robinson scheme; the latter enforces a fundamental restriction on the speed of correlations spreading in nonrelativistic quantum systems~\cite{lieb1971finite, hastings2004lieb, roberts2016lieb,  lashkari2013towards}. In Eq.~\eqref{eq:def_thermal_OTOC}, it is convenient to consider pairs of operators $V,W$ which at $t=0$ act nontrivially on different subsystems, thus, commute; we follow this convention here.

The commutator $C_{V,W}(t)$ is intimately linked to the out-of-time-order correlator (OTOC)~\cite{larkin1969quasiclassical,kitaev2015simple} which is a four-point function with an unconventional time-ordering
%\footnote{In our definition, $F_{V,W}(t)$ is real and vanishes for commuting operators.}
\begin{align}
F_{V,W} (t) \coloneqq   \Tr \left( V^\dagger(t) W^\dagger V(t) W \rho_\beta  \right)  .
\end{align}
The connection between the two arises when $V,W$ are unitary; Eq.~\eqref{eq:def_thermal_OTOC} then immediately reduces to $C_{V,W}(t) = 1- \Real \left[ F_{V,W}(t) \right]$. In this Letter we focus on the infinite-temperature, $\beta = 0$ case.

Through the years, several key signatures of quantum chaos~\cite{fishman1982chaos, adachi1988quantum, gutzwiller1990chaos, haake2013quantum} have been introduced. The initial exponential growth of the OTOC was proposed as a diagnostic of quantum chaos~\cite{maldacena2016bound, roberts2015diagnosing, polchinski2016spectrum, mezei2017entanglement, huang2017out, zhang2019information, roberts2017chaos, prakash2020scrambling}. However, a careful analysis has revealed that information scrambling does not always necessitate chaos~\cite{PhysRevB.98.134303,PhysRevLett.123.160401,luitz2017information,PhysRevE.101.010202,PhysRevLett.124.140602,hashimoto2020exponential}.

Per se, the OTOC's ability to probe dynamical features clearly depends on the choice of operators $V,W$. However, it is desirable to be able to capture these features as independently as possible from the specific choice of operators. This insensitivity can be achieved by averaging over a set of operators, a strategy also considered in Refs.~\cite{cotler2017chaos, roberts2017chaos, fan2017out, de2019spectral, ma2020early, touil2020quantum, yan2020information}. It is crucial to remark that, for the averaged OTOC to faithfully capture information spreading, the averaging process must \textit{preserve the initial locality of the system}, i.e., which subsystems $V,W$ initially act upon --- an observation that was quintessential in revealing the correct behavior of the OTOC and its connection with Loschmidt echo~\cite{yan2020information}.

Given a bipartition of a finite-dimensional Hilbert space $\mathcal H = \mathcal H_A \otimes \mathcal H_B \cong \mathbb C^{d_A} \otimes \mathbb C^{d_B}$, we will henceforth focus on averaging $C_{V_A,W_B}(t)$ over the (independent) unitary operators $V_A$ and $W_B$, whose support is over subsystems $A$ and $B$, respectively. The resulting quantity
\begin{align} \label{eq:defnition_G}
G(t) \coloneqq 1 - \frac{1}{d} \Real \! \int dV dW \Tr \big( V_A^\dagger(t) W_B^\dagger V_A(t) W_B  \big),
\end{align} 
depends only on the dynamics and the Hilbert space cut, where we denote $V_A = V \otimes I_B$, $W_B = I_A \otimes W$ and the averaging is performed according to the Haar measure~\cite{watrous2018theory}. We will refer to $G(t)$ for brevity as the \textit{bipartite OTOC}, and analyzing its properties will be the focus of the present Letter.

%$G(t)$ was recently introduced in Ref.~\cite{yan2020information}. It has been recently shown in Ref.~\cite{yan2020information} that $G(t)$, under the assumptions of Markovianity and weak coupling between $A$ and $B$, exhibits a close connection with the Loschmidt echo~\cite{peres1984stability}.

%, with the purpose of bridging OTOC and the Loschmidt echo~\cite{peres1984stability}. It was shown, under the assumptions of Markovianity and weak coupling between $A$ and $B$,

It was recently shown in Ref.~\cite{yan2020information}, where $G(t)$ was first introduced, under the assumptions of (i)~weak coupling between $A$ and $B$, and (ii)~Markovianity, that $G(t)$ exhibits a close connection with the Loschmidt echo~\cite{peres1984stability,jalabert2001environment}; the latter has been widely employed to characterize chaos~\cite{gorin2006dynamics,goussev2012loschmidt}. Here, we first show, without any of the previous assumptions, that $G(t)$ is, in fact, amenable to exact analytical treatment, and we uncover its direct relation with entropy production, information spreading, and entanglement. We also rigorously prove that the average case is also the typical one, hence justifying the averaging process. Our main results are stated in the theorems that follow. All proofs of the claims appearing in the text can be found in Appendix~\ref{sec:app:proofs}.

%Moreover, its late-time behavior is intimately connected to quantum entanglement, and finally reveal its direct relation with entropy entropy production and information spreading via concrete formulas. All proofs of the claims appearing in the text can be found in Appendix~\ref{sec:app:proofs}.

\prlsection{The bipartite OTOC} We begin by bringing $G(t)$ in a more explicit form which will be the starting point for a sequence of results. This can be achieved by working on the doubled space $\mathcal H \otimes \mathcal H'$, where $\mathcal H' = \mathcal H_{A'} \otimes \mathcal H_{B'}$ is a replica of the original Hilbert space.

\begin{restatable}{theo}{GAverageUnitaries} \label{prop:GAverageUnitaries}
Let $S_{AA'}$ be the operator over $\mathcal H \otimes \mathcal H'$ that swaps $A$ with its replica $A'$ and $d = \dim(\mathcal H)$. Then
\begin{align} \label{eq:G_main}
G(t) = 1 - \frac{1}{d^2} \Tr \left( S_{AA'}  U_t^{\otimes 2}  S_{AA'} U_t^{\dagger \otimes 2} \right). 
\end{align}
The analogous expression for $ BB'$ also holds.
\end{restatable}

The above formula immediately exposes a connection between the bipartite OTOC and the \textit{operator entanglement} of the evolution $ E_{\mathrm {op}}(U_t)$, as defined in Ref.~\cite{zanardi2001entanglement} (see also Appendix~\ref{sec:app:proofs} for the relevant definitions). The two quantities, remarkably, coincide exactly. This observation also allows one to express the \textit{entangling power}~\cite{zanardi2000entangling} $e_{\mathrm{P}}(U_t)$ as a function of the bipartite OTOC for the symmetric case $d_A = d_B$. The former quantifies the average entanglement produced by the evolution and has been established as an indicator of global chaos in few-body systems~\cite{wang2004entanglement, lakshminarayan2001entangling, scott2003entangling, pal2018entangling}. 

\begin{restatable}{theo}{EntanglingPower}  \label{prop:entangling_power}
Let $G_U$ denote the bipartite OTOC for the evolution~$U$. Then, (i)~$E_{\mathrm {op}}(U_t) = G_{U_t}$, and (ii)~for a symmetric bipartition $d_A = d_B$,
\begin{align} \label{eq:entangling_power}
    e_{\mathrm{P}}(U_t) = \frac{d}{(\sqrt{d} + 1)^2} \left( G_{U_t} + G_{U_t  S_{AB}} - G_{S_{AB}} \right).
\end{align}
\end{restatable}

For the finite-temperature case, Eq.~\eqref{eq:G_main} admits a straightforward generalization which we report in Appendix~\ref{sec:app:proofs}. However, a direct connection with operator entanglement and entangling power may not be so simple.

{\em How informative is the average $G(t)$?}---~ Usually, one is interested in behavior of the OTOC for a typical choice of random unitary operators.
%, instead of the average case we have so far considered.
Because of measure concentration~\cite{ledoux2001concentration}, we prove that the two essentially coincide; i.e., the probability that a random instance deviates significantly from the mean is exponentially suppressed as the dimension of either of the subsystems $A$ and $B$ grows large.

\begin{restatable}{prop}{GDoubleConcentration} \label{prop:GDoubleConcentration}
Let $P(\epsilon)$ be the probability that a random instance of $C_{V_A,W_B}(t)$ deviates from its Haar average $G(t)$ more than $\epsilon$. Then,
\begin{align} \label{eq:double_concentration}
P(\epsilon) \le 2 \exp\left( -\frac{\epsilon^2 d_{\max}}{64} \right),
\end{align}
where $d_{\max} = \max \{ d_A, d_B \}$.
\end{restatable}

In the definition of the bipartite OTOC and to obtain the replica formula Eq.~\eqref{eq:G_main},
we have so far considered averaging over the uniform (Haar) ensemble which continuously extends over the whole unitary group. Although natural from a mathematical viewpoint, this choice can turn out to be rather complicated on physical and numerical grounds~\cite{emerson2003pseudo}. Nonetheless, we show in Appendix~\ref{sec:app:measures} that Haar averaging can be replaced by any unitary ensemble that forms a 1-design~\cite{divincenzo2002quantum, renes2004symmetric, scott2006tight, gross2007evenly} without altering $G(t)$. Such ensembles mimic the Haar randomness only up to the first moment, which is the depth of randomness that the OTOC can probe~\cite{roberts2017chaos}. The latter assumption is thus much weaker than Haar randomness. For instance, consider the case of a spin-$1/2$ many-body system split into two parts, $A$ and $B$. Instead of averaging over Haar random unitaries $V_A$ and $W_B$, that typically do not factor, the 1-design (equivalent) picture prescribes to instead consider only fully factorized unitaries with support over $A$ and $B$, e.g., products of local Pauli matrices.

\prlsection{Time-averaging the bipartite OTOC} In finite-dimensional quantum systems, nontrivial quantum expectation values or quantities such as $C_{V,W}(t)$ do not converge to a limit for $t \to \infty$. Instead, after a long time they typically oscillate around an equilibrium value~\cite{reimann2008foundation,linden2009quantum,
venuti2011exact,nahum2018operator,
daug2019detection,alhambra2020time,
daug2019detection} which can be extracted by time-averaging $\overline{X(t)} \coloneqq \lim_{T \to \infty} \frac{1}{T} \int_0^T dt \, X(t)$. We now turn to examine this long-time behavior $\overline {G(t)}$ of the bipartite OTOC  as a function of the Hamiltonian and the Hilbert space cut.

Let us begin with the case of a chaotic dynamics, which entails level repulsion statistics~\cite{haake2013quantum} and an ``incommensurable'' relation among the energy levels. As such, chaotic Hamiltonians satisfy (either exactly or to very good approximation) the no-resonance condition (NRC): The energy levels and energy gaps feature nondegeneracy. This has important implications for the long-time behavior of their bipartite OTOC, as we will see soon.

Let us spectrally decompose $H = \sum_k E_k \ket{\phi_k} \!\bra{\phi_k}$ and use $\rho^{(\chi)}_k \coloneqq \tr_{\overline \chi} \left( \ket{\phi_k} \!\bra{\phi_k}\right)$  to denote the reduced density operator over $\chi = A,B$ corresponding to the $k$th Hamiltonian eigenstate ($\overline \chi$ corresponds to the complement). Below, $\braket{X,Y} \coloneqq \tr (X^\dagger Y )$ denotes the Hilbert-Schmidt inner product~\cite{bhatia2013matrix}, which gives rise to the operator 2-norm $\left\| X \right\|_2 \coloneqq \sqrt{\braket{X,X}}\,$.
\begin{restatable}{prop}{Rmatrix} \label{prop:R_matrix}
Consider a Hamiltonian satisfying the NRC.
% and let $\{ \rho^{(\chi)}_k \}_{k=1}^d$ be the collection of quantum states over $\chi = A,B$ resulting after tracing out. 
Then
\begin{align} \label{eq:G_ave_R_matrix}
\overline{G(t)}^\mathrm{NRC} = 1- \frac{1}{d^2} \sum_{\chi \in \{ A,B \}} \Big(  \big\| R^{(\chi)} \big\|_2^2 - \frac{1}{2} \big\| R^{(\chi)}_D \big\|_2^2 \Big)
\end{align}
where $R^{(\chi)}$ is the Gram matrix of the reduced Hamiltonian eigenstates $\{ \rho^{(\chi)}_k \}_{k=1}^d$, i.e.,
\begin{align} 
R_{kl}^{(\chi)} \coloneqq \braket{\rho^{(\chi)}_k , \rho^{(\chi)}_l}
\end{align}
while $\big( R^{(\chi)}_D \big)_{kl} \coloneqq R^{(\chi)}_{kl} \delta_{kl}$.
\end{restatable}

%The above expression in terms of the Gram matrix has a few convenient mathematical properties.
Let us first point out some basic, yet important properties of the above formula. The matrix $R^{(\chi)}$ is real and symmetric, while $R_D^{(\chi)}$ is positive semidefinite and diagonal. 
%and hence Eq.~\eqref{eq:G_ave_R_matrix} is a function of the eigenvalues of $R^{(\chi)}$ and $R^{(\chi)}_D$.
Moreover, the completeness of the Hamiltonian eigenvectors imposes $\sum_k \rho^{(\chi)}_k = d_{\overline \chi} I$; thus the rescaled $\tilde R^{(\chi)} \coloneqq R^{(\chi)} / d_{\overline \chi}$ are doubly stochastic, i.e., $\sum_i \tilde R^{(\chi)}_{ij}$ = $\sum_i \tilde R^{(\chi)}_{ji} = 1 $ $\forall j$. As $\tilde R^{(\chi)}$ is a (rescaled) Gram matrix, its eigevalues are non-negative, upper bounded by 1, and at most $d_{\chi}^2$ of them are nonzero~\cite{bhatia2013matrix}. This last property follows from the fact that $\Rank\tilde R^{(\chi)} = \dim \Span \{ \rho_k^{(\chi)} \}_k \le d_{\chi}^2$. Observe also that $\big\| R^{(A)}_D \big\|_2^2 = \big\| R^{(B)}_D \big\|_2^2$ as two states $\rho_k^{(A)}$ and $\rho_k^{(B)}$ always have the same spectrum (up to irrelevant zeroes).

\prlsection{Bipartite OTOC and entanglement} \autoref{prop:R_matrix} makes it possible to bridge the long-time behavior of the bipartite OTOC with the entanglement structure of the Hamiltonian eigenstates. Let us begin with the symmetric case where $d_A = d_B$ and all $\ket{\phi_k}$ are maximally entangled with respect to the $A$-$B$ Hilbert space cut. This limit uniquely determines the time average for the NRC case, regardless of the exact Hamiltonian eigenbasis. In general, however, knowledge of the entanglement is not enough to uniquely determine the equilibration value; the inner products $R^{(\chi)}_{kl}$ go beyond probing just the spectrum of the reduced states. A simple substitution in Eq.~\eqref{eq:G_ave_R_matrix} gives for the maximally entangled case $\overline{G_{\mathrm{ME}}(t)}^\mathrm{NRC} = (1 - 1/d)^2$. We will later show the upper bound $G(t) \le 1 - 1/d^2_{\min}$; therefore the equilibrium value for the bipartite OTOC in this case is nearly maximal, as expected for highly entangled models (e.g.,~\cite{huang2019finite,harrow2019separation}).

How robust is this conclusion for chaotic Hamiltonians with a possibly asymmetric bipartition? Typical eigenstates of chaotic Hamiltonians, as also predicted by the eigenstate thermalization hypothesis~\cite{deutsch1991quantum, srednicki1994chaos, rigol2008thermalization}, are believed to obey a volume law for the entanglement entropy. Moreover, their entanglement properties in the bulk resemble those of Haar random pure states~\cite{d2016quantum,huang2019universal,lu2019renyi}.
%, which typically have nearly maximal entanglement~\cite{d2016quantum,huang2019universal,lu2019renyi}.
We will now show that high entanglement for the Hamiltonian eigenstates necessarily implies that the deviation of the actual equilibration value from $\overline{G_{\mathrm{ME}}(t)}^\mathrm{NRC}$ is small.

%Let us therefore examine the difference $\big| \overline{G_{\mathrm{ME}}(t)}^\mathrm{NRC} \!\!- \overline{G(t)}^\mathrm{NRC} \big|$ as a function entanglement.

%Eigenstates of chaotic Hamiltonians possess generically a large amount of entanglement. 
%Let us examine what happens when the bipartition size varies, as a function of the entanglement. High equilibration values will persist for Hamiltonians whose eigenstates possess a large, but not necessarily maximal, amount of entanglement; this is the relevant case for chaotic Hamiltonians. 

%What is arguably more interesting is studying the persistence of this conclusion for various bipartition sizes and Hamiltonian eigenstates possessing a large, but not necessarily maximal, amount of entanglement. This is the relevant case for chaotic Hamiltonians. As also predicted by the eigenstate thermalization hypothesis~\cite{deutsch1991quantum,srednicki1994chaos,rigol2008thermalization}, typical eigenstates in chaotic systems are believed to obey a volume law for the entanglement entropy. Moreover, their entanglement properties in the bulk resemble those of (Haar) random pure states, which typically have nearly maximal entanglement~\cite{d2016quantum,huang2019universal,lu2019renyi}. 

It is convenient for this purpose to quantify the amount of entanglement via the linear entropy~\cite{horodecki2009quantum,bose2000mixedness} of the reduced state $E(\ket{\psi_{AB}}) \coloneqq S_{\mathrm{lin}}\left( \tr_{ \chi}\ket{\psi_{AB}}\!\bra{\psi_{AB}} \right)$, where $S_{\mathrm{lin}} (\rho) \coloneqq 1 - \tr(\rho^2)$. The latter will also emerge naturally later when we express the bipartite OTOC in terms of entropy production. Notice that $ E \le 1 - 1/d_{\max} \coloneqq E_{\max}$, which is achievable only for $d_A = d_B$.

\begin{restatable}{prop}{boundNRC} \label{prop:bound_NRC}
If $E_{\max} - E(\ket{\phi_k}) \le \epsilon$ holds for at least a fraction $\alpha$ of the Hamiltonian eigenstates, then $\big| \overline{G_{\mathrm{ME}}(t)}^\mathrm{NRC} - \overline{G(t)}^\mathrm{NRC} \big| \le \alpha J + (1-\alpha)K$, where
\begin{subequations}
\begin{align}
J &\coloneqq   \frac{6 \epsilon }{d_{\min}}   + \frac{5 \epsilon^2}{2}   +  2 \frac{\lambda^2-1}{d_{\max}^2} \\
K  &\coloneqq   \left( 1 + \frac{2}{d_{\min}} \right)(1-\alpha) + \frac{2}{d} + 4 (\epsilon + \sqrt{\epsilon})  
\end{align}
\end{subequations}
and $\lambda = d_{\max} / d_{\min}$.
\end{restatable}

%\begin{restatable}{prop}{boundNRC} \label{prop:bound_NRC}
%If the entanglement of the Hamiltonian eigenstates deviates up to $\epsilon$ from $E_{\max}$ with respect to the $A$-$B$ cut, i.e.,  $E_{\max} - E(\ket{\phi_k}) \le \epsilon$ for all $k$, then
%\begin{align} \label{eq:time_avg_entanglement}
%\big| \overline{G_{\mathrm{ME}}(t)}^\mathrm{NRC} \!\!- \overline{G(t)}^\mathrm{NRC} \big| \le \frac{6 \epsilon }{d_{\min}}   + \frac{5 \epsilon^2}{2}  \! +  2 \frac{\lambda^2-1}{d_{\max}^2}
%\end{align}
%where $\lambda = d_{\max} / d_{\min}$.
%\end{restatable}
%
%
The above bound provides a sufficient condition such that  the bipartite OTOC equilibrates around $\overline{G_{\mathrm{ME}}(t)}^\mathrm{NRC}$. It is expressed in terms of the fraction $\alpha$ of the highly entangled eigenstates, their entanglement and the asymmetry of the $A$-$B$ bipartition. Notice that the bound simplifies considerably for the case $\alpha = 1$ and $d_{\mathrm{min}}=d_{\mathrm{max}} = \sqrt{d}$, that is, \(\big| \overline{G_{\mathrm{ME}}(t)}^\mathrm{NRC} - \overline{G(t)}^\mathrm{NRC} \big| \leq \epsilon (6/\sqrt{d} + 5 \epsilon/2)\) which should hold to a good approximation for Hamiltonians with high entanglement in the bulk of the energies. Applied to chaotic Hamiltonians\footnote{Here chaoticity concretely means that the Hamiltonian spectrum satisfies the NRC and that the entanglement of the typical eigenvectors in the bulk, which determine the equilibration value, resembles that of Haar random vectors~\cite{lubkin1978entropy, hamma2012quantum}, i.e., $\Tr \left( \rho_\chi ^2 \right) \approx (d_A + d_B)/(d+1) $ thus $\epsilon = O(1/d_{\min})$ and $\alpha \approx 1$.}, the bound of \autoref{prop:bound_NRC} indicates that the bipartite OTOC will equilibrate near $\overline{G_{\mathrm{ME}}(t)}^\mathrm{NRC}$, with deviations up to $O(1 /d_{\min}^2)$. For a fixed ratio $\lambda$ and as $d$ grows, $\overline{G(t)}^\mathrm{NRC}$ hence converges to $\overline{G_{\mathrm{ME}}(t)}^\mathrm{NRC}$ for all chaotic systems. Since $G(t) \le 1 - 1/d_{\min}^2$, fluctuations around the time average are necessarily insignificant, justifying the term equilibration.
%In Appendix~\ref{sec:app:proofs} we also present a generalization of \autoref{prop:bound_NRC}, where only a given fraction of the Hamiltonian eigenstates is assumed to have entanglement $\epsilon$-close to $E_{\max}$.

\prlsection{Beyond chaotic Hamiltonians} 
We now relax the ``strong'' level repulsion, i.e., NRC, criterion and uncover how a hierarchy of constraints, each implying a different strength of chaos, is reflected in the equilibration value of the bipartite OTOC.

Integrable models, which possess a structured spectrum, are expected to violate the NRC. Nevertheless, notice that Eq.~\eqref{eq:G_ave_R_matrix}, although derived under the NRC, can still be evaluated for an (arbitrary) choice of orthonormal eigenvectors of the Hamiltonian. We will refer to the resulting value as the \textit{NRC estimate} of the time average and we will shortly show that this estimate always constitutes an upper bound of the actual equilibration value (and coincides with it for chaotic Hamiltonians). This is of both conceptual and practical importance, as evaluating the NRC estimate is considerably less intensive than calculating the exact value.

In fact, one can make a broader claim. For that, we first sketch three types of averaging processes over $G$, increasingly shifting away from the strong chaoticity limit. Each of them gives rise to a corresponding estimate for the (exact) equilibration time-average value $\overline{G(t)}$. (i)~$\overline{G}^{\mathrm{Haar}}$: Averaging over (global) Haar random unitary operators $U \in U(d)$ in place of the time evolution. This averaging process is ``beyond chaos'', in the sense that it does not conserve energy, in contrast with time averaging over any Hamiltonian evolutions. Its estimate (only a function of the dimension) is given later in Eq.~\eqref{eq:Haar_average}.
(ii)~$\overline{G(t)}^{\mathrm{NRC}}$: Time-average, assuming the Hamiltonian has nondegenerate energy levels and nondegenerate energy gaps. The corresponding estimate is Eq.~\eqref{eq:G_ave_R_matrix}. (iii)~$\overline{G(t)}^{\mathrm{NRC}^+}$: As before, but assuming the Hamiltonian may have degenerate spectrum, but the energy gaps (between the different levels) are nondegenerate. Its estimate depends only on the eigenprojectors of the Hamiltonian and can be found in Appendix~\eqref{sec:app:proofs}.
%(iv)~\textit{General Hamiltonian evolution}: Ordinary time-averaging of the evolution without any assumption.

The value of the Haar average can be performed exactly, with the 	result
\begin{align} \label{eq:Haar_average}
\overline{G}^{\mathrm{Haar}} = \frac{(d_A^2-1)(d_B^2 - 1)}{d^2 - 1}\,.
\end{align}

The following ordering holds.
\begin{restatable}{theo}{NRCUpperBound} \label{prop:NRC_upper_bound}
For any given Hamiltonian, the corresponding estimates are related with the exact time average $ \overline{G(t)}$ as
\begin{align} \label{eq:comparison_time_avgs_ineq}
\overline{G}^{\mathrm{Haar}} \ge \overline{G(t)}^{\mathrm{NRC}} \ge \overline{G(t)}^{\mathrm{NRC}^+} \ge \overline{G(t)} \;.
\end{align}
\end{restatable}

The above constitutes a proof that coincidences in the spectrum of a Hamiltonian up to the ``gaps of gaps'' (i.e., degeneracy over the energy levels and their gaps) always \textit{reduce} the equilibration value of the bipartite OTOC.

\begin{figure}[t]
\includegraphics[width=\columnwidth]{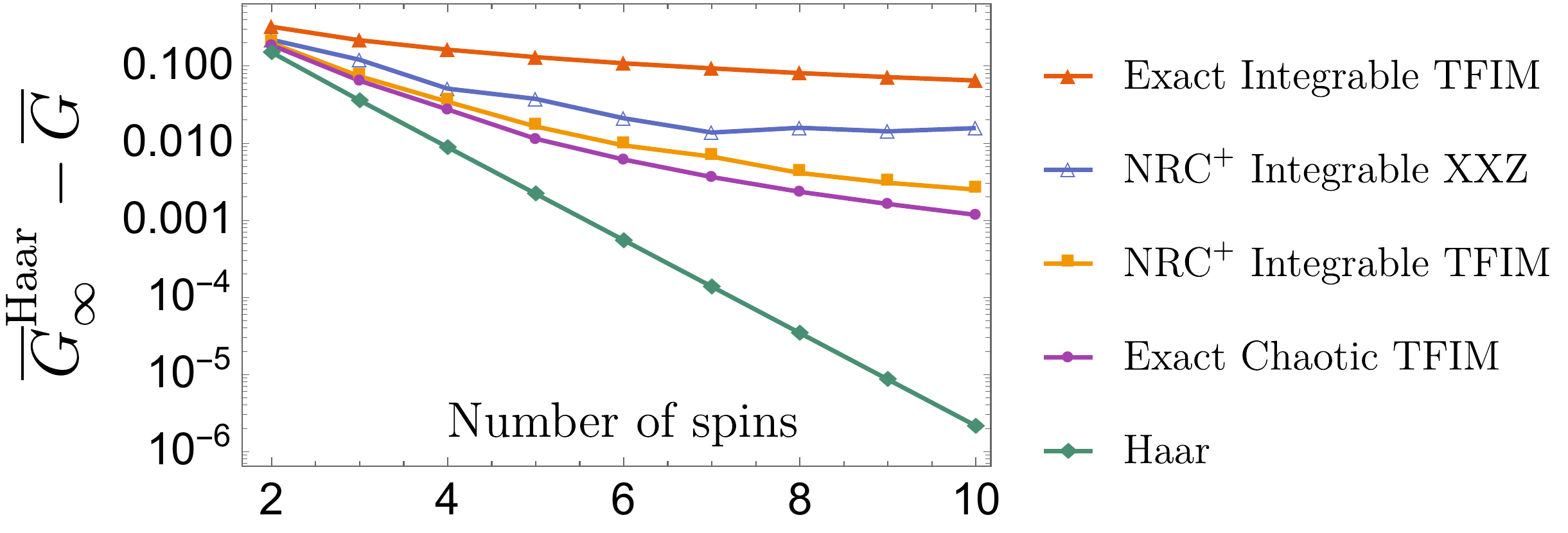}
\caption{\label{fig:numerics}Logarithmic plot of various $\overline{G}$ estimates, along with the exact time average, for fixed $d_A = 2$ as a function of the total number of spins $n$. $\overline{G}^{\mathrm{Haar}}_\infty = 3/4$ corresponds to the Haar estimate for $n \to \infty$. For the chaotic phase of the TFIM ($g=-1.05$, $h = 0.5$), the NRC constitutes a satisfactory, though imperfect, approximation. The chaotic and integrable phases ($h = 0$) can be clearly distinguished through the equilibration behavior of the bipartite OTOC. 
For the integrable XXZ model (we set $J = 0.4$, $\Delta = 2.5$), the NRC\textsuperscript{+} estimate coincides (up to numerical error) with the exact time average. Inequality~\eqref{eq:comparison_time_avgs_ineq} holds valid in all cases.}
\end{figure}

% \begin{figure}[t]
% \includegraphics[width=\columnwidth]{combined.pdf}
% \caption{\label{fig:numerics}Left: Plot of $\overline{G}$ with fixed $d_A = 2$ as a function of the total number of spins $n$ for $H_{\mathrm{I}}$ ($g=1$, periodic-boundary conditions). NRC fails to capture the correct value due to the highly structured spectrum in this integrable model. Right: Similarly for $H_{\mathrm{XXZ}}$ ($J = 0.4$, $\Delta = 2.5$, open-boundary conditions). In the chaotic phase $g=-1.05$, $h = 0.5$ all estimates (except Haar) coincide, up to numerical error. In the integrable phase ($g= h = 0$) NRC\textsuperscript{+} coincides (up to numerical error) with the exact result. Eq.~\eqref{eq:comparison_time_avgs_ineq} holds valid in all cases.}
% \end{figure}

Let us now numerically compare each of the estimates for two models of spin-1/2 chains with open-boundary conditions: (i) transverse-field Ising model (TFIM) with nearest neighbour interaction, $H_{\mathrm{I}} = - \sum_i(\sigma_i^z \sigma_{i+1}^z + g \sigma_i^x + h \sigma_i^z)$ (ii) nearest-neighbor XXZ interaction $H_{\mathrm{XXZ}} = - J \sum_i  ( \sigma_i^x \sigma_{i+1}^x + \sigma_i^y \sigma_{i+1}^y + \Delta \sigma_i^z \sigma_{i+1}^z)$. Recall that $H_{\mathrm{I}}$ for $h=0$ is integrable in terms of free-fermions, while $H_{\mathrm{XXZ}}$ by Bethe Ansatz techniques. The two types of solutions yield qualitatively different spectra; free fermion solutions necessarily violate nondegeneracy of the gaps. This is reflected in the accuracy of the estimates (see~\autoref{fig:numerics}). Although the NRC estimate provides essentially the exact equilibration values for the chaotic phase of the TFIM, it overestimates them in the integrable phase. On the other hand, NRC\textsuperscript{+} is essentially exact for the integrable case of the $H_{\mathrm{XXZ}}$ due to the lack of coincidences in the gaps. The results obtained here corroborate existing studies in the literature, where the (short- and) long-time behavior of the OTOC was studied for various many-body systems; see Refs.~\cite{PhysRevE.100.042201,PhysRevLett.121.210601,PhysRevLett.121.124101}.

\prlsection{Bipartite OTOC and subsystem evolution} We have so far focused on examining the behavior of the bipartite OTOC from the perspective of closed systems, i.e., over the full bipartite Hilbert space $\mathcal H_A \otimes \mathcal H_B$. One can instead express $G(t)$ as a function of the reduced time dynamics over only either $\mathcal H_A$ or $\mathcal H_B$ (and the corresponding duplicate), at the expense of giving up unitarity. This can be easily realized by formally performing a partial trace in Eq.~\eqref{eq:G_main}, which immediately results in the following equivalent expression for the bipartite OTOC.
\begin{restatable}{prop}{ReducedDynamics} \label{prop:reduced_dynamics}
Let $\Uplambda_t ^{\!(A)} (\rho_A) \!\coloneqq\! \tr_B \!\left[ U_t \!\left( \rho_A \otimes \dfrac{I_B}{d_B} \right)  U^\dagger _t \right]$ be the reduced dynamics over A when the environment B is initialized in a maximally mixed state. Then,
\begin{align} \label{eq:G_reduced_dynamics}
G(t) = 1 - \frac{1}{d_A^2}  \tr \left[  S_{AA'} \big( \Uplambda_t^{\!(A) }\big)^{ \otimes 2} (S_{AA'})\right] .
\end{align} 
The analogous expression for $BB'$ also holds.
\end{restatable}

The quantum map $\Uplambda^{\!(\chi)}_t$ is unital; i.e., the maximally mixed state is a fixed point. As such, the transformation $\rho_\chi \mapsto \Uplambda^{\!(\chi)}_t (\rho_\chi)$ results always in an output state whose spectrum is more disordered than the input one~\cite{bengtsson2017geometry}.
%, as captured precisely by the majorization statement $\rho_\chi \succ \Uplambda^{\!(\chi)}_t (\rho_\chi)$.
As a result, when $\rho_\chi$ is pure, the effect of the reduced time dynamics is to scramble and, hence, produce entropy. Let us now turn to examine this connection more closely.

\prlsection{Bipartite OTOC as entropy production}
We now show that the bipartite OTOC $G(t)$ is nothing but a measure of the average entropy production over pure states, with the latter quantified by linear entropy $S_{\mathrm{lin}}$.

\begin{restatable}{theo}{EntropyProduction} \label{prop:entropy_production}
\begin{align} \label{eq:entropy_production}
G(t) = \frac{d_\chi + 1}{d_{\chi}} \int dU \, S_{\mathrm{lin}} \left[ \Uplambda^{\!(\chi)}_t  (\ket{\psi_U} \! \bra{\psi_U}) \right]
\end{align}
where $\chi = A,B$ and $\ket{\psi_U} \coloneqq U \ket{\psi_0}$ corresponds to Haar random pure states over $\mathcal H_\chi$.
\end{restatable}

%The strictest of the two bounds is thus for $\chi$ being the smallest between $A,B$ and is achievable only when $ \Uplambda^{\!(\chi)}_t$ is equal to the completely depolarizing map $\mathcal T^{(\chi)} (\cdot) \coloneqq \Tr (\cdot) \dfrac{I_\chi}{d_\chi}$.

In this manner, the bipartite OTOC can be fully characterized by linear entropy measurements over any of the $A,B$ subsystems. To obtain a satisfactory estimate of the mean in the rhs of Eq.~\eqref{eq:entropy_production}, one does not, in practice, need to sample over the full Haar ensemble. An adequate estimate can be obtained with a rapidly decreasing number of necessary samples, as the dimension $d_{\chi}$ grows. More precisely, let $\tilde P(\epsilon)$ be the probability of the entropy  $S_{\mathrm{lin}} \big[ \Uplambda^{\!(\chi)}_t  \big( \ket{\psi}\! \bra{\psi} \big) \big]$ deviating from $\frac{d_{\chi} }{d_{\chi}+1} G(t)$ more than $\epsilon$ for an instance of a random state. We show in Appendix~\ref{sec:app:proofs} that
\begin{align} \label{eq:concentration_linear_entropy}
\tilde P(\epsilon) \le \exp \left(  - \frac{d_\chi \epsilon^2}{64} \right) .
\end{align}

The linear entropy, although, per se, a nonlinear functional, can be turned into an ordinary expectation value if two (uncorrelated) copies of the quantum state are simultaneously available, $1 - S_{\mathrm{lin}} = \tr \left(S \rho ^{\otimes 2} \right)$ for $S= S_{AA'} S_{BB'}$. This fact can be exploited to simplify its experimental accessibility~\cite{ekert2002direct, bovino2005direct,PhysRevLett.93.110501,daley2012measuring, islam2015measuring}. More recently, protocols based on correlating measurements over random bases have also been developed to measure entropies~\cite{brydges2019probing,elben2019statistical,huang2020predicting,elben2020mixed}, as well as OTOCs~\cite{vermersch2019probing,joshi2020quantum}. As a result, \autoref{prop:entropy_production} and the typicality result Eq.~\eqref{eq:concentration_linear_entropy} suggest that the bipartite OTOC is, in turn, tractable via linear entropy measurements. We provide more details in Appendix~\ref{sec:app:linear_entropy}.

From Eq.~\eqref{eq:entropy_production} one can also infer the upper bound $G(t) \le 1 - 1/d_{\chi}^2 \coloneqq G_{\max}^{(\chi)}$ announced earlier that follows from the range of the linear entropy function. The bound is thus achievable only when $ \Uplambda^{\!(\chi)}_t$ is equal to the completely depolarizing map $\mathcal T^{(\chi)} (\cdot) \coloneqq \Tr (\cdot) \dfrac{I_\chi}{d_\chi}$.

Finally, we remark that linear entropy occurs rather naturally in relation with the bipartite OTOC, as demonstrated by \autoref{prop:entangling_power} (where it lies implicitly in the definition of operator entanglement and entangling power) and \autoref{prop:entropy_production}. This fact has its roots in the definition of the OTOC, which is intimately related to the Frobenious norm. Relevant relations for the linear entropy have been also reported in \cite{fan2017out}. Starting from the inequality $S_{\mathrm{lin}} (\rho) \le S (\rho) $ between the linear and von Neumann entropies ($S(\rho) \coloneqq - \tr [ \rho \log (\rho) ] $), one can also obtain the corresponding estimates for the latter.

\prlsection{Bipartite OTOC and information spreading} The bipartite OTOC measures the average ability of the reduced time evolution to erase information, as captured by the entropy production over a random pure state. This naturally raises the question as to whether $G(t)$ can also be understood as a measure of distance between  $\Uplambda^{\!(\chi)}_t $ and the depolarizing map $\mathcal T^{(\chi)}$, that is, in the space of quantum channels (i.e., Completely Positive and Trace Preserving (CPTP) maps~\cite{nielsen2002quantum}).

A straightforward answer can be obtained by resorting to the duality between quantum states and operations~\cite{nielsen2002quantum}. Let $\rho_{\mathcal E} \coloneqq \mathcal E \otimes \mathcal I (\ket{\phi^+} \! \bra{\phi^+})$ denote the (Choi) state corresponding to the CPTP map $\mathcal E$, where $\ket{\phi^+} \coloneqq d^{-1/2} \sum_{i=1}^d \ket{ii}$ is a maximally entangled state.
% (of appropriate dimensions).

\begin{restatable}{prop}{Choi} \label{prop:Choi}
The bipartite OTOC is a measure of the distance between the reduced time evolution and the depolarizing map:
\begin{align} \label{eq:Choi}
G(t) = G_{\max}^{(\chi)} - \big\|  \rho_{\Uplambda^{\!(\chi)}_t} - \rho_{\mathcal T^{(\chi)}}  \big\|_2^2 \,.
\end{align}
\end{restatable}

As an application, the proposition above can be utilized to bound the distance $\big\| \Uplambda^{\!(\chi)}_t - \mathcal T^{(\chi)} \big\|_{\lozenge}\,$ given by the diamond norm~\cite{kitaev1997quantum,kitaev2002classical}; the latter is a well-established measure of distance between quantum channels\footnote{Bounding the difference in terms of the quantum processes also constraints the distinguishability in terms of states: $ {\big\| \mathcal E_1 (\rho)  - \mathcal E_2 (\rho) \big\|_1} \le \big\| \mathcal E_1 - \mathcal E_2  \big\|_\lozenge$for all states and quantum processes.} since it admits an operational interpretation in terms of discrimination on the level of quantum processes~\cite{wilde2013quantum}. The distinguishability of the two operations satisfies $ \big\| \Uplambda^{\!(\chi)}_t - \mathcal T^{(\chi)} \big\|_{\lozenge} \le d_{\chi}^{3/2} \sqrt{ G_{\max}^{(\chi)}  -  G(t)} $ (see Appendix~\ref{sec:app:proofs}); therefore if $G_{\max}^{(\chi)}  -  G(t)$ decays faster than $d_\chi^{-3}$, then asymptotically the two channels are essentially indistinguishable.

%\prlsection{Summary} We have studied the average OTOC and its typicality for random unitary operators supported over a bipartition. We investigated its late-time behavior and established a link with the entanglement of Hamiltonian eigenstates. We have considered a hierarchy of estimates for the equilibration value of the bipartite OTOC, showed that coincidences in the Hamiltonian spectrum can only monotonically alter this value, and numerically illustrated their accuracy. We have provided an operational significance to the bipartite OTOC by revealing its correspondence with entangling power, operator entanglement, production of linear entropy in each of the subsystems, and also considered its connection with information scrambling on the level of quantum channels. Applying further these theoretical results to concrete many-body models constitutes a direction for future research.

\prlsection{Summary} We showed that the bipartite OTOC is amenable to exact analytical treatment and, quite remarkably, is equal to the operator entanglement of the dynamics. This identity allows one to establish a rigorous quantitative connection between the OTOC and the notion of entangling power, a well-established quantifier of few-body chaos. This may provide insights into recent work involving ``dual-unitaries'' and many-body chaos~\cite{akila2016particle,PhysRevLett.123.210601,PhysRevB.101.094304,10.21468/SciPostPhys.8.4.067}; the latter maximize operator entanglement~\cite{10.21468/SciPostPhys.8.4.067,PhysRevLett.125.070501}. We then turned to late-time averages of the bipartite OTOC and provided a hierarchy of estimates for systems that violate the conditions of a ``generic spectrum''. Finally, we unraveled the operational significance of the OTOC by establishing intimate connections with entropy production and information scrambling at the level of quantum channels. Possible future directions include applying further these theoretical tools to concrete many-body systems and uncovering relations with thermalization, localization, and other many-body phenomena.

%The out-of-time-order correlator (OTOC) probes the propagation of information between pairs of operators. Since many qualitative features of the OTOC are insensitive to the specific choice of operators, as long as their locality is fixed, it constitutes a meaningful simplification to focus on random averages. Given a bipartition of the Hilbert space, we analytically perform the uniform average over pairs of random unitary operators supported over the two disjoint regions. Our result reveals a precise connection between the OTOC, entropy production and the effectiveness of the reduced evolution to erase information. We determine the OTOC equilibration value for chaotic Hamiltonians, connect with entanglement, and prove that any coincidences in the Hamiltonian spectrum can only monotonically alter this value.

%\textit{New summary.}--- Provide analytical estimates for the bipartite OTOC. Justify averaging with concentration of measure. OTOC equilibration value reveals connection with eigenstate entanglement. Entropy production as an operational significance. Hierarchy of constraints. Numerical simulations. Information scrambling and channnel distinguishability. Conjecture/future directions. 

%\acknowledgments

\prlsection{Acknowledgments} G.S. is thankful to N.A.~Rodr\'{i}guez~Briones for the interesting discussions and to the Louisa house in Kitchener for the hospitality.  Research was funded by the Deutsche Forschungsgemeinschaft (DFG, German Research Foundation) under Germany's Excellence Strategy -- EXC-2111 -- 390814868. P.Z. acknowledges partial support from the National Science Foundation Grant No.~PHY-1819189. Research was sponsored by the Army Research Office and was accomplished under Grant No.~W911NF-20-1-0075. The views and conclusions contained in this document are those of the authors and should not be interpreted as representing the official policies, either expressed or implied, of the Army Research Office or the U.S. Government. The U.S. Government is authorized to reproduce and distribute reprints for Government purposes notwithstanding any copyright notation herein.

%\prlsection{Acknowledgements} 

\bibliography{refs}

\onecolumngrid
\widetext
\newpage

\appendix

\section*{\large Appendices}

\section{Proofs} \label{sec:app:proofs}

Here we restate the Theorems and Propositions, as well as other mathematical claims appearing in the main text, and give their proof.

\subsection*{Theorem~\ref{prop:GAverageUnitaries}}

\GAverageUnitaries*
\begin{proof}
Let $S$ be the operator  over $\mathcal H \otimes \mathcal H'$ that swaps $\mathcal H$ with its replica $\mathcal H'$. Then for any operators $X,Y$ acting over $\mathcal H$ it holds that
\begin{align} \label{eq:app:swap_trick}
\Tr \left( X Y \right) = \Tr \left[ S (X \otimes Y) \right] ,
\end{align}
as it can be easily verified by expressing both sides in a basis.
Notice that in our case, where $\mathcal H$ carries a bipartition, one can further decompose $S = S_{A A'} S_{B B'}$.

Using the above identity the OTOC averaging in Eq.~\eqref{eq:defnition_G} can be written as
\begin{align*}
G(t) &= 1 - \frac{1}{d} \Real \int dV dW \Tr \left(S \, V_A^\dagger(t) W_B^\dagger \otimes  V_A(t) W_B  \right) \\
 & =  1 - \frac{1}{d} \Real \int dV dW \Tr \left(S  U_t^{\dagger \otimes 2} (V_A^\dagger \otimes  V_A)  U_t^{\otimes 2}  (W_B^\dagger \otimes  W_B)  \right) \\
 & = 1 - \frac{1}{d} \Real \Tr \left[ S U_t^{\dagger \otimes 2} \left( \int dV V_A^\dagger \otimes  V_A   \right) U_t^{\otimes 2} \left( \int dW W_B^\dagger \otimes  W_B  \right) \right] .
\end{align*}
Now the two independent averages can be easily performed since for unitary operators over $\mathcal H \cong \mathbb C^d$ the corresponding Haar integrals evaluate to
\begin{align} \label{eq:app:u_avg_vectorized}
\int dU U \otimes U^\dagger = \frac{S}{d}
\end{align}
where $S$ is again the swap operator over the doubled space.

A quick way to prove the well-known identity~\eqref{eq:app:u_avg_vectorized} is by using Eq.~\eqref{eq:app:swap_trick} to write
\begin{align*}
U X U^\dagger = \tr_{\mathcal H'} \left[ (U \otimes U^\dagger)( X \otimes I) S \right]
\end{align*}
and then using the fact that
\begin{align}
\int dU U X U^\dagger = \frac{\Tr(X)}{d}
\end{align}
which follows directly from the left/right invariance of the Haar measure~\cite{watrous2018theory}.

Using Eq.~\eqref{eq:app:u_avg_vectorized} twice, we get
\begin{align*}
G(t) &= 1 - \frac{1}{d} \Real \Tr \left( S U_t^{\dagger \otimes 2} \frac{S_{AA'}}{d_A} U_t^{\otimes 2} \frac{S_{BB'}}{d_B} \right) \\
 & = 1 - \frac{1}{d^2} \Tr \left( S_{AA'}  U_t^{\otimes 2}  S_{AA'} U_t^{\dagger \otimes 2} \right) . 
\end{align*}
Since $\left[ S , X^{\otimes 2} \right] = 0$ for all operators $X$, the analogous expression for $BB'$ holds, i.e.,
\begin{align}
G(t) = 1 - \frac{1}{d^2} \Tr \left( S_{BB'}  U_t^{\otimes 2}  S_{BB'} U_t^{\dagger \otimes 2} \right).
\end{align}
\end{proof}

Notice that the symmetry of the Haar measure forces the bipartite OTOC to be time reversal invariant, i.e., $G(t) = G(-t)$.

Finally, we also note that that there is a straightforward generalization of \autoref{prop:GAverageUnitaries} to any finite temperature thermal state. Following similar steps as above, one gets for for the thermal version of the bipartite OTOC
\begin{align}
G(t) = 1 - \frac{1}{d} \Real \Tr \left( (\rho_\beta \otimes I_{A'B'}) U_t^{\dagger\, \otimes 2} S_{AA'}  U_t^{ \otimes 2}  S_{AA'}\right).
\end{align}

\subsection*{Theorem~\ref{prop:entangling_power}}

\EntanglingPower*

Before giving the proof, let us first recall the definitions of operator entanglement~\cite{zanardi2001entanglement} and entangling power~\cite{zanardi2000entangling}.

The main idea behind operator entanglement is to first express the unitary evolution $U$ (over the bipartite Hilbert space $\mathcal H_{AB}$) as a state in the doubled space $\mathcal H_{AB}\otimes \mathcal H_{A'B'}$ via
\begin{align}
    \ket{U} =  U \otimes I_{A'B'} \ket{\phi^+}
\end{align}
for the maximally entangled state $\ket{\phi^+} = \frac{1}{\sqrt{d}} \sum_{i=1}^d \ket{i}_{AB}\ket{i}_{A'B'}$ and then evaluate the linear entropy of the state $\sigma_U = \tr_{BB'} \left(  \ket{U} \! \bra{U} \right)$, i.e.,
\begin{align} \label{eq:app_operator_entanglement}
    E_{\mathrm{op}} (U) \coloneqq S_{\mathrm{lin}} (\sigma_U) = 1 - \Tr (\sigma_U^2) . 
\end{align}

The entangling power~\cite{zanardi2000entangling} of a quantum evolution $U$ over a bipartite quantum system $\mathcal H = \mathcal H_A \otimes \mathcal H_B$ is defined as the average entanglement that the evolution generates when acting on random separable pure states. More specifically,
\begin{align}
    e_{\mathrm{P}} (U) \coloneqq \int dV dW E \left[   U \left( \ket{\psi_V}_A \ket{\psi_W}_B \right) \right],
\end{align}
where $\ket{\psi_V}_A = V \ket{\psi_0}_A$ corresponds to Haar random pure states over $A$ ($\ket{\psi_0}_A$ is an irrelevant reference state), and similarly for B, while $E(\ket{\psi_{AB}}) \coloneqq S_{\mathrm{lin}}\left( \tr_{ B}\ket{\psi_{AB}}\!\bra{\psi_{AB}} \right)$ is the entanglement of the resulting state, as measured by the linear entropy.

\begin{proof}
\textbf{(i)}~The key observation here is that the bipartite OTOC $G_{U_t}$, in the form of  Eq.~\eqref{eq:G_main}, coincides with the operator entanglement $E(U_t)$ as defined in Ref.~\cite{zanardi2001entanglement} (see Eq.~(6) therein).
Evaluating the expression~\eqref{eq:app_operator_entanglement}, as in the proof of \autoref{prop:GAverageUnitaries}, one obtains exactly Eq.~\eqref{eq:G_main}, hence $E_{\mathrm{op}} (U_t) = G_{U_t}$.

\textbf{(ii)} For the symmetric case $d_A = d_B$, the result follows by combining the first part of the current Theorem and Eq.~(12) of Ref.~\cite{zanardi2001entanglement}.

Finally, we note that by direct substitution, one has $G_{S_{AB}} = 1 - 1/d$.
\end{proof}

\subsection*{Proposition~\ref{prop:GDoubleConcentration}}

\GDoubleConcentration*

The proof relies on measure concentration and, in particular, Levy's lemma which we shall recall shortly (see, e.g.,~\cite{anderson2010introduction}). Below we are also going use various operator (Schatten) $k$-norms~\cite{bhatia2013matrix}; the latter are defined as $\left\| X  \right\|_k \coloneqq \left(  \sum_i s^k_i  \right)^{1/k}$ where $\{ s_i \}_i$ are the singular values of $X$. The case $\left\| X \right\|_\infty \coloneqq \max_i \left\{  s_i \right\}_i$ corresponds to the usual operator norm. For  $k \ge l$, one always has $\| X \|_k \le \| X \|_l$.

We also remind the reader that a function $f: U(d) \to \mathbb R$ is said to be Lipschitz continuous with constant $K$ if it satisfies
\begin{align}
\left| f(V) - f(W)  \right| \le K \left\| V - W \right\|_2
\end{align}
for all $V,W \in U(d)$. For brevity, in this section we denote the Haar averages as $\braket{(\cdot)}_U$ and also occasionally drop the explicit time dependence.

\begin{theorem*}[Levy's lemma] Let $U \in U(d)$ be distributed according to the Haar measure and $f: U(d) \to \mathbb R$ be a Lipschitz continuous function. Then for any $\epsilon > 0$
\begin{align}
\Prob\{ \left| f(U) - \braket{f(U)}_U \right| \ge \epsilon \} \le \exp \left( - \frac{d \epsilon^2}{4 K^2} \right) ,
\end{align}
where $K$ is a Lipschitz constant.
\end{theorem*}

During the course of the proof of \autoref{prop:GDoubleConcentration}, the following two continuity results will come in handy.

\begin{lemma} \label{lemma:app:lipschitz}
\begin{enumerate}[(i)]
\item The function $f_{W}(V) : U(d_A) \to \mathbb R$ with $f_{W}(V) \coloneqq C_{V_A,W_B}(t) $ is Lipschitz continuous with constant $K_f = 2$ for all $t \in \mathbb R$ and $W \in U(d_B)$.
\item The function $g(W) : U(d_B) \to \mathbb R$ with $g(W) \coloneqq \braket{C_{V_A,W_B}(t)}_{V} $ is Lipschitz continuous with constant $K_g = 2/d_A$ for all $t \in \mathbb R$.
\end{enumerate}
\end{lemma}

\begin{proof}[Proof of lemma]
\textbf{(i)} Let $X,Y \in U(d_A)$. We need to show that
\begin{align*}
\left| f_W(X) - f_W (Y)  \right| \le K_f \left\| X - Y  \right\|_2.
\end{align*}

Following the proof of \autoref{prop:GAverageUnitaries}, we can express 
\begin{align*}
f_W(V) = 1 - \frac{1}{d} \Real \tr \left[ S  U_t^{\dagger \otimes 2} (V_A^\dagger \otimes  V_A)  U_t^{\otimes 2}  (W_B^\dagger \otimes  W_B)    \right]
\end{align*}
therefore
\begin{align*}
\left| f_W(X) - f_W (Y)  \right| & \le  \frac{1}{d} \left| \tr \left[ U_t^{ \otimes 2}   (W_B^\dagger \otimes  W_B) S U_t^{\dagger \otimes 2}  (X_A^\dagger \otimes  X_A - Y_A^\dagger \otimes  Y_A)  \right]  \right| \\
& \le \frac{1}{d} \big\| X_A^\dagger \otimes  X_A - Y_A^\dagger \otimes  Y_A \big\|_1 ,
\end{align*}
where in the last step we used the inequality $\left\| \Tr \left( A B \right) \right\| \le \left\| A  \right\|_1 \left\|  B \right\|_\infty$ and the fact that $\big\| U_t^{ \otimes 2}   (W_B^\dagger \otimes  W_B) S U_t^{\dagger \otimes 2} \big\|_\infty= 1$ since the operator within the norm is unitary.

In order to express the last norm as a function of the difference $X_A - Y_A$, we first add and subtract $Y^\dagger_A \otimes X_A$ and then use the triangle inequality. This  results in
\begin{align*}
\frac{1}{d} \big\| X_A^\dagger \otimes  X_A - Y_A^\dagger \otimes  Y_A \big\|_1 & \le \frac{1}{d} \left( \big\| (X_A^\dagger -Y_A^\dagger) \otimes X_A \big\|_1 + \big\| Y_A^\dagger \otimes (X_A - Y_A)  \big\|_1  \right) \\
& \le \frac{1}{d} \left(  \big\| X_A^\dagger - Y_A^\dagger  \big\|_\infty   \big\| I \otimes X_A \big\|_1 + \big\| X_A - Y_A  \big\|_\infty   \big\| Y_A^\dagger \otimes I  \big\|_1 \right)
\end{align*}
where for the last step we utilized the inequality $\left\| A B \right\|_1 \le \left\| A \right\|_1  \left\| B \right\|_\infty$.  Now notice that $\big\| I \otimes X_A \big\|_1  = d$ since $X_A$ is unitary, and similarly for $\big\| Y_A^\dagger \otimes I  \big\|_1$. Therefore we can bound
\begin{align*}
\left| f_W(X) - f_W (Y)  \right|  \le \big\| X_A - Y_A \big\|_\infty + \big\| X^\dagger_A - Y^\dagger_A \big\|_\infty \le  2 \big\| X_A - Y_A \big\|_\infty = 2 \big\| X - Y \big\|_\infty  \le 2 \big\| X - Y \big\|_2 ,
\end{align*}
from which clearly one can take $K_f = 2$.
\\ \\
\textbf{(ii)} First notice that the Haar average over $V_A = V \otimes I_B$ can be performed, as was done in the proof of \autoref{prop:GAverageUnitaries}. The result is
\begin{align*}
g(W) &= 1 - \frac{1}{d} \Real \Tr \left[ S U_t^{\dagger \otimes 2} \frac{S_{AA'}}{d_A} U_t^{\otimes 2}  W_B^\dagger \otimes  W_B   \right] \\
 & = 1 - \frac{1}{d} \Real \Tr \left[ U_t^{\dagger \otimes 2} \frac{S_{BB'}}{d_A} U_t^{\otimes 2}  W_B^\dagger \otimes  W_B   \right] .
\end{align*}
Considering the relevant difference, we can bound
\begin{align*}
\left| g(X) - g (Y)  \right| & \le \frac{1}{d_A} \frac{1}{d} \left|  \Tr \left[  U_t^{\dagger \otimes 2} S_{BB'} U_t^{\otimes 2}  (X_B^\dagger \otimes  X_B  -  Y_B^\dagger \otimes  Y_B ) \right]  \right| \\
& \le  \frac{1}{d_A} \frac{1}{d} \big\| X_B^\dagger \otimes  X_B - Y_B^\dagger \otimes  Y_B \big\|_1 .
\end{align*}
Now one can follow the exact same steps as in part (i); the result is identical except of the extra factor $1/d_A$ that carries through, which originates from the averaging. This results in
\begin{align*}
\left| g(X) - g (Y)  \right| \le \frac{2}{d_A}  \big\| X - Y \big\|_2
\end{align*}
from which one can take $K_g = 2/d_A$.
\end{proof}

Everything is now in place to give the proof of \autoref{prop:GDoubleConcentration}. 

\begin{proof}
Let $\epsilon > 0$. We want to show that, for $V \in U(d_A)$ and $W \in U(d_B)$ distributed independently according to the Haar measure, it holds
\begin{align*}
\Prob \left( \gamma \ge \epsilon \right) \le \exp\left( -\frac{\epsilon^2 d_{\max}}{64} \right)
\end{align*}
where $\gamma \coloneqq \left| C_{V_A,W_B} - G  \right|$ and by definition $G = \braket{C_{V_A,W_B}}_{V,W}$.

Let us consider any pair $V_A, W_B$ that satisfies  $\epsilon \le \gamma$. Then, from the triangle inequality also
\begin{align*}
\epsilon \le \alpha + \beta,
\end{align*}
where we set $\alpha  \coloneqq \left| C_{V_A,W_B} - \braket{C_{V_A,W_B}}_{V}  \right| $  and 
$\beta  \coloneqq \left| \braket{C_{V_A,W_B}}_{V} - G \right|$.
Hence we have for the corresponding probabilities
\begin{align*}
\Prob \left\{ \gamma \ge \epsilon \right\} \le \Prob \left\{ \alpha + \beta \ge \epsilon \right\} .
\end{align*}
However, if $\alpha + \beta \ge \epsilon$ then necessarily  $\alpha \ge \epsilon/2$ or $\beta \ge \epsilon/2$, therefore we also have
\begin{align*}
\Prob \left\{ \alpha + \beta \ge \epsilon \right\} \le \Prob \left( \left\{ \alpha  \ge \epsilon / 2 \right\} \cup \left\{ \beta  \ge \epsilon / 2 \right\} \right).
\end{align*}
Using the standard union bound over the last expression results in
\begin{align} \label{eq:app:ineq_abc}
\Prob \left\{ \gamma \ge \epsilon \right\} \le  \Prob  \left\{ \alpha  \ge \epsilon / 2 \right\} + \Prob   \left\{ \beta  \ge \epsilon / 2 \right\}.
\end{align}

The two Probabilities in Eq.~\eqref{eq:app:ineq_abc} can be bounded using Levy's lemma. For that, let us first define the auxiliary functions $f_{W}(V)$ and $g(W)$ as in \autoref{lemma:app:lipschitz}. Combining the Lipschitz continuity result from there with Levy's lemma, one gets measure concentration bounds
\begin{subequations} \label{eq:app:levy_bound_fg}
\begin{align}
\Prob_V \{ \left| C_{V_A,W_B} - \braket{C_{V_A,W_B}}_{V}  \right| \ge \epsilon/2 \} &\le \exp \left( - \frac{d_A \epsilon^2}{64} \right) \quad \forall W  \label{eq:app:levy_bound_f_only} \\
\Prob \{ \braket{C_{V_A,W_B}}_{V} - G \ge \epsilon/2 \} &\le \exp \left( - \frac{d_A^2 d_B \epsilon^2}{64} \right)
\end{align} 
\end{subequations}
We are almost done; it suffices to notice that the bound~\eqref{eq:app:levy_bound_f_only} is uniform in $W$, hence it is also applicable to $\Prob  \left\{ \alpha  \ge \epsilon / 2 \right\} $. Therefore we arrive at
\begin{align}
\Prob \{ \left| C_{V_A,W_B}(t) - G(t)  \right| \ge \epsilon \} \le  \exp \left( - \frac{d_A \epsilon^2}{64} \right) + \exp \left( - \frac{d_A^2 d_B \epsilon^2}{64} \right) \le 2  \exp \left( - \frac{d_A \epsilon^2}{64} \right)  .
\end{align}
Notice the resulting bound is independent of the dynamics, as long as the latter is unitary. Finally, one can obtain the analogous bound for $A \leftrightarrow B$ by inverting the roles of $V$ and $W$ in the proof. Therefore we obtain Eq.~\eqref{eq:double_concentration}.
\end{proof}

\subsection*{Proposition~\ref{prop:R_matrix}}

\Rmatrix*

Here we give a straightforward proof assuming the NRC holds exactly. For a more detailed discussion, see also the section of the proof of \autoref{prop:NRC_upper_bound}.

\begin{proof}
Our starting point is Eq.~\eqref{eq:G_main}, which we need to time average. Since the Hamiltonian is  by assumption nondegenerate, we can spectrally decompose $H = \sum_{k=1}^d E_k P_k$, where $P_k \coloneqq \ket{\phi_k} \! \bra{\phi_k}$. We then have
\begin{align*}
\overline{G(t)}^\mathrm{NRC} = 1 - \frac{1}{d^2} \sum_{klmn} \overline{ \exp \big[ i (E_k + E_l -E_m - E_n) t \big]  } \tr \left[ S_{AA'} (P_k \otimes P_l) \, S_{AA'} \, (P_m \otimes  P_n)  \right] .
\end{align*}
Time averaging the exponential results in
\begin{align*}
\overline{ \exp \big[ i (E_k + E_l -E_m - E_n) t \big]  } = \delta_{E_k + E_l -E_m - E_n,0} \stackrel{\mathclap{\scriptsize \mbox{NRC}}}{=\joinrel=} \delta_{k,m} \delta_{l,n} + \delta_{k,n} \delta_{l,m} - \delta_{k,l} \delta_{l,m} \delta_{m,n}
\end{align*}
where in the last step we used the fact that energy gaps are nondegenerate. Thus
\begin{align*}
\overline{G(t)}^\mathrm{NRC} & =1  -  \frac{1}{d^2} \Big( \sum_{kl}  \tr \left[ S_{AA'}    (P_k \otimes P_l) \, S_{AA'} \, (P_k \otimes  P_l)   \right]  +   \sum_{kl}  \tr \left[ S_{AA'}    (P_k \otimes P_l) \, S_{AA'} \, (P_l \otimes  P_k)   \right]   \\ & \hspace{0.55 \columnwidth} - \sum_{k}  \tr \left[ S_{AA'}    (P_k \otimes P_k) \, S_{AA'} \, (P_k \otimes  P_k)   \right]   \Big) \\
& = 1 -   \frac{1}{d^2} \Big( \sum_{kl} \big|  \tr   \left[(P_k \otimes P_l) \, S_{AA'}   \right] \big|^2 +  \sum_{kl} \big|  \tr   \left[(P_k \otimes P_l) \, S_{BB'}   \right] \big|^2 -  \sum_{k} \big|  \tr   \left[(P_k \otimes P_k) \, S_{AA'}   \right] \big|^2   \big) ,
\end{align*}
where for the second term we used that $P_l \otimes P_k = S (P_k \otimes P_l) S$ and $S = S_{AA'} S_{BB'}$.

Now, notice that partial traces can be formally performed, giving
\begin{align*}
 \tr _{AA'BB'}  \left[(P_k \otimes P_l) S_{AA'}   \right]  = \tr_{AA'} \left[  \tr_{BB'} (P_k \otimes P_l) S_{AA'} \right] = \tr_{AA'} \left[ (\rho_k^{(A)} \otimes \rho_l^{(A')} ) S_{AA'} \right] = \tr\left( \rho_k^{(A)} \rho_l^{(A)} \right) = R_{kl}^{(A)}, 
\end{align*}
and similarly
\begin{align*}
\tr _{AA'BB'}  \left[(P_k \otimes P_l) \, S_{BB'}   \right] &= R_{kl}^{(B)} \\
\tr _{AA'BB'}  \left[(P_k \otimes P_k) \, S_{AA'}   \right] &=\tr _{AA'BB'}  \left[(P_k \otimes P_k) \, S_{BB'}   \right] = R_{kk}^{(A)} = R_{kk}^{(B)}
\end{align*}
where in the last line we used the fact that the spectra of $\rho_k^{(A)}$ and $\rho_k^{(B)}$ are equal, up to (irrelevant for the trace) zeroes. The result follows by expressing the matrix 2-norm as $\left\| X \right\|_2^2 = \sum_{ij} \left| X_{ij} \right|^2$.
\end{proof}

\subsection*{Proposition~\ref{prop:bound_NRC}}

Before proceeding with the proof, let us briefly comment on the need of including the parameter $\alpha$, which corresponds to the fraction of the highly entangled eigenstates of the Hamiltonian. For certain Hamiltonian models (e.g., the class of gapped, local Hamiltonians over one-dimensional lattice systems) it is well known that the ground state follows an area law for the entanglement entropy~\cite{eisert2010colloquium}. Thus for larger system sizes $\epsilon$ cannot be chosen to be small for the ground state (and also possibly for the low lying excited states), even for the symmetric $d_A = d_B$ bipartition. Nevertheless, in the bulk of the spectrum, typical eigenstates are expected to obey instead a volume law, which is compatible with an $\epsilon$ that can be chosen to be suitably small. Therefore, we expect that, for certain physically relevant models, a large fraction $\alpha$ can be assumed to satisfy this condition.

\boundNRC*

\begin{proof}

To simplify the notation, we assume $d_A \le d_B$. Let us also define $I = \{ k: E_{\max} - E(\ket{\phi_k}) \le \epsilon \}$, i.e., $I$ is the index set of those Hamiltonian eigenstates that deviate at most by $\epsilon$ from $E_{\max}$, while we use $\bar {I}$ to label the rest of the eigenstates. By assumption, $\left| I  \right| \ge \alpha d$.

First of all, notice that one can express the difference $E_{\max} - E(\ket{\psi_{AB}})$ as the distance
\begin{align*}
E_{\max} - E(\ket{\psi_{AB}}) = \Tr ( \rho_B^2 ) - 1/d_{B } = \big\|  \rho_B - I/d_B \big\|_2^2 \ge \big\|  \rho_A - I/d_A \big\|_2^2 = \Tr ( \rho_A^2 ) - 1/d_A \,\;.
\end{align*}
Setting for brevity $\Delta_k^{(\chi)} \coloneqq  \rho_k^{(\chi)} - I/d_\chi$ ($\chi = A,B$),
we have for all $k \in I$ that $E_{\max} - E(\ket{\phi_{k}}) = \big\|  \Delta_k^{(B)} \big\|_2^2 \le \epsilon$ and hence also $\big\|  \Delta_k^{(A)} \big\|_2^2 = \big\|  \rho_k^{(A)} - I/d_A \big\|_2^2 \le \epsilon $. It will be convenient for later to express
\begin{align} \label{eq:app:deltas_aux}
\big| \braket{\rho_k^{(\chi)},\rho_l^{(\chi)}} \big|^2 = \big| \braket{I / d_\chi + \Delta_k^{(\chi)},  I / d_\chi + \Delta_l^{(\chi)}} \big|^2 = \big| \frac{1}{d_\chi} + \braket{\Delta_k^{(\chi)} , \Delta_l^{(\chi)}} \big|^2 = \frac{1}{d_\chi^2} + \frac{2}{d_\chi} \braket{\Delta_k^{(\chi)} , \Delta_l^{(\chi)}} + \braket{\Delta_k^{(\chi)} , \Delta_l^{(\chi)}}^2.
\end{align}
Moreover, by the Cauchy-Schwartz inequality, 
\begin{subequations} \label{eq:app_estimates_Delta}
\begin{align}
\big| \braket{\Delta_k^{(\chi)} , \Delta_l^{(\chi)}} \big| \le \big\| \Delta_k^{(\chi)} \big\|_2 \big\| \Delta_l^{(\chi)} \big\|_2 
\end{align}
while
\begin{align}
\big\| \Delta_k^{(\chi)} \big\|_2 ^ 2 \le \begin{cases}
\epsilon  & \mbox{if } k \in I , \\
1 - \frac{1}{d_\chi} & \mbox{otherwise.}
\end{cases}
\end{align}
\end{subequations}

Let's start from Eq.~\eqref{eq:G_ave_R_matrix}. Using the fact that $\big\| R^{(A)}_D \big\|_2^2 = \big\| R^{(B)}_D \big\|_2^2$ and recalling $\overline{G_{\mathrm{ME}}(t)}^\mathrm{NRC} = (1 - 1/d)^2$ we get by the triangle inequality
\begin{align} \label{eq:app_three_terms}
\big| \overline{G_{\mathrm{ME}}(t)}^\mathrm{NRC} - \overline{G(t)}^\mathrm{NRC}  \Big| & \le   \Big| \frac{1}{d^2} \big\|     R^{(A)} \big\|_2^2 -  \frac{1}{d}  \Big| +  \Big| \frac{1}{d^2} \big\|    R^{(B)} \big\|_2^2 -  \frac{1}{d}  \Big|   + \frac{1}{d^2}  \big| \big\|  R_D^{(A)} \big\|_2^2 -  1 \big| .
\end{align}
To bound the first term we write
\begin{align*}
   \Big| \frac{1}{d^2} \big\|     R^{(A)} \big\|_2^2 -  \frac{1}{d}  \Big| = \Big| \frac{1}{d^2} \sum_{kl} \big| \braket{\rho_k^{(A)},\rho_l^{(A)}} \big|^2  - \frac{1}{d} \Big|  \le \frac{1}{d_A^2} - \frac{1}{d} 
+ \frac{1}{d^2} \sum_{kl} \left( \frac{2}{d_A} \big| \braket{\Delta_k^{(A)} , \Delta_l^{(A)}} \big| + \braket{\Delta_k^{(A)} , \Delta_l^{(A)}}^2 \right)   
\end{align*}
where we used Eq.~\eqref{eq:app:deltas_aux}. Splitting both of the sums as $\sum_{k} = \sum_{k \in I} + \sum_{k \notin I}$ and using Eqs.~\eqref{eq:app_estimates_Delta} we have
\begin{align*}
\frac{1}{d^2} \sum_{kl}  \big| \braket{\Delta_k^{(A)} , \Delta_l^{(A)}} \big| & \le  \epsilon \alpha^2 + 2 \alpha(1-\alpha) \sqrt{\epsilon \left( 1 - \frac{1}{d_A} \right)}  + (1 - \alpha)^2 \left(1 - \frac{1}{d_A}  \right)
\end{align*}
and
\begin{align*}
\frac{1}{d^2} \sum_{kl} \braket{\Delta_k^{(A)} , \Delta_l^{(A)}} ^2 & \le  \epsilon^2 \alpha^2 + 2 \alpha(1-\alpha) \epsilon \left( 1 - \frac{1}{d_A} \right)  + (1 - \alpha)^2 \left(1 - \frac{1}{d_A}  \right)^2 .
\end{align*}
Putting them together, and relaxing some inequalities for clarity, we obtain for the first term of Eq.~\eqref{eq:app_three_terms}
\begin{align*}
\Big| \frac{1}{d^2} \big\|     R^{(A)} \big\|_2^2 -  \frac{1}{d}  \Big| \le \frac{1}{d_A^2} - \frac{1}{d} + \alpha \epsilon \left(\frac{2}{d_A} + \epsilon \right)   + (1-\alpha)^2 (1 + \frac{2}{d_A})  + 2 (1-\alpha) (\epsilon + \sqrt{\epsilon}).
\end{align*}
Analogously for the second term of Eq.~\eqref{eq:app_three_terms},
\begin{align*}
   \Big| \frac{1}{d^2} \big\|     R^{(B)} \big\|_2^2 -  \frac{1}{d}  \Big|  \le   \frac{1}{d} - \frac{1}{d_B^2} + \alpha \epsilon \left(\frac{2}{d_A} + \epsilon \right)   + (1-\alpha)^2 (1 + \frac{2}{d_A})  + 2 (1-\alpha) (\epsilon + \sqrt{\epsilon}) .
\end{align*}
For the third one, we have
\begin{align*}
  \big\| R_D^{(A)} \big\|_2^2 = \sum_k \big| \braket{\rho_k^{(A)} , \rho_k^{(A)}} \big|^2  =  \frac{d_B}{d_A} + \frac{2}{d_A} \sum_k \braket{\Delta_k^{(A)} , \Delta_k^{(A)}} + \sum_k \braket{\Delta_k^{(A)} , \Delta_k^{(A)}}^2 .
\end{align*}
Using similar manipulations as above, and under the convention $d_A \le d_B$,
\begin{align*}
\frac{1}{d^2}  \big| \big\|  R_D^{(A)} \big\|_2^2 -  1 \big|  \le  \frac{1}{d^2} \left( \frac{d_B}{d_A} - 1 \right) + \frac{1}{d} \left[ \alpha \left( \frac{2 \epsilon}{d_A} + \epsilon^2 \right) + (1-\alpha) \left( \frac{2}{d_A} +1 \right)  \right]
\end{align*}

Putting the inequalities together, we have
\begin{multline}
\big| \overline{G_{\mathrm{ME}}(t)}^\mathrm{NRC} - \overline{G(t)}^\mathrm{NRC}  \big| \le \\  \frac{\lambda - 1 }{d^2} + \frac{\lambda^2 - 1}{d_B^2} + \alpha \left[ 2 \epsilon \left( \frac{2}{d_A}  + \frac{1}{d_A^2 d_B} \right)   +  \epsilon^2 \left( 2 + \frac{1}{d} \right) \right] + (1 - \alpha ) \left[ 2 (1-\alpha) \left( 1+ \frac{2}{d_A} \right) + \frac{2}{d} + 4\left( \epsilon + \sqrt{\epsilon} \right) \right]
\end{multline}
which can be relaxed to give the final result by using $\dfrac{\lambda^2 - 1}{d_B^2} \ge \dfrac{\lambda - 1 }{d^2}$.
\end{proof}

\subsection*{Theorem~\ref{prop:NRC_upper_bound}}

\NRCUpperBound*

Before giving the proof of the Theorem, we first briefly discuss some general facts regarding infinite time averages, their connection with the NRC and the NRC\textsuperscript{+}, and how they give rise to the corresponding estimates.

Let us consider unitary quantum dynamics $\mathcal U_t (\cdot) = U_t (\cdot) U^\dagger _t $ generated by a Hamiltonian $H = \sum_k \tilde E_k \Pi_k$, where $\Pi_k$ denotes the projector onto the $k\textsuperscript{th}$ eigenspace. As a warm-up, let us calculate the time average of the superoperator $\mathcal U_t$. The latter can be easily performed by noticing that $ \overline{\exp\big[-i (\tilde E_k - \tilde E_l) t \big]} = \delta_{kl}$. It results to
\begin{align}
\mathcal P_{ H} \coloneqq \overline{\mathcal U_t}  = \sum_k \Pi_k (\cdot) \Pi_k
\end{align}
which is the (Hilbert-Schmidt orthogonal) projector onto the commutant of the algebra generated by $\{ \Pi_k \}_k$, i.e., the projector whose range is the space of operators commuting with $H$.

The object of interest for us is, in fact, $\overline{\mathcal U_t^{\otimes 2}}$ since
\begin{align}
\overline{G(t)} = 1- \frac{1}{d^2} \braket{S_{AA'} , \overline{\mathcal U_t^{\otimes 2}} (S_{AA'})}.
\end{align}
Reasoning as above, it follows that the resulting superoperator is again a projector, whose range is the space of operators over the replicated Hilbert space $\mathcal H^{\otimes 2}$ that commute with $H^{(2)} \coloneqq H \otimes I + I \otimes H$. The projector can be explicitly expressed as
\begin{align}
\mathcal P_{ H^{(2)}} \coloneqq \, \overline{\mathcal U_t^{\otimes 2}} \,   =  \sum_{klmn } \delta_{\tilde E_k - \tilde E_m , \tilde E_l - \tilde E_n} \Pi_k \otimes \Pi_l (\cdot) \Pi_m \otimes \Pi_n 
%  = \sum_{\mathclap{\substack{klmn \\ \text{such that}  \\ \tilde E_k + \tilde E_l = \tilde E_m + \tilde E_n}}}  \Pi_k \otimes \Pi_l (\cdot) \Pi_m \otimes \Pi_n .
\end{align}
To evaluate the above sum, let us for a moment examine what happens when the energy gaps $\{\tilde E_k - \tilde E_l \}_{kl}$ are nondegenerate. i.e.,
\begin{align}
\mathrm{NRC}^+: \qquad \tilde E_k + \tilde E_l = \tilde E_m + \tilde E_n  \Longleftrightarrow ( k = m  \ \wedge \, l = n ) \, \vee \, ( k = n  \ \wedge \, l = m ) .
\end{align}
We will refer to this condition over the spectrum as $\mathrm{NRC}^+$, since it constitutes a relaxed version of the $\mathrm{NRC}$.  Without any assumption over the spectrum, one can always separate two contributions
\begin{align}
\mathcal P_{\mathcal H^{(2)}} = \mathcal P_{\mathrm{NRC}^+} + \mathcal P_{\overline{\mathrm{NRC}^+}} 
\end{align} 
where
\begin{align}
\mathcal P_{\mathrm{NRC}^+} \coloneqq \sum_{kl }  \Pi_k \otimes \Pi_l (\cdot) \Pi_k \otimes \Pi_l + \sum_{k l}  \Pi_k \otimes \Pi_l (\cdot) \Pi_l \otimes \Pi_k - \sum_k  \Pi_k \otimes \Pi_k (\cdot) \Pi_k \otimes \Pi_k 
\end{align}
and $\mathcal P_{\overline{\mathrm{NRC}^+}} $ is any possibly remaining piece, which vanishes if and only if the Hamiltonian does indeed satisfy $\mathrm{NRC}^+$.

Disregarding $\mathcal P_{\overline{\mathrm{NRC}^+}} $, one gets the estimate
\begin{align}
 \overline{G(t)}^{\mathrm{NRC}^+} &\coloneqq 1 - \frac{1}{d^2} \tr\left[ S_{AA'} \mathcal P_{\mathrm{NRC}^+} \left( S_{AA'} \right)  \right]\\
& = 1  -  \frac{1}{d^2} \Big( \sum_{kl}  \tr \left[ S_{AA'}    (\Pi_k \otimes \Pi_l) \, S_{AA'} \, (\Pi_k \otimes  \Pi_l)   \right]  +   \sum_{kl}  \tr \left[ S_{AA'}    (\Pi_k \otimes \Pi_l) \, S_{AA'} \, (\Pi_l \otimes  \Pi_k)   \right] \nonumber  \\ & \hspace{0.45 \columnwidth} - \sum_{k}  \tr \left[ S_{AA'}    (\Pi_k \otimes \Pi_k) \, S_{AA'} \, (\Pi_k \otimes  \Pi_k)   \right]   \Big) , \label{eq:app:NRCp_expanded}
\end{align}
where the second equation follows from the proof of \autoref{prop:R_matrix}. Clearly, if all projectors $\{\Pi_k \}$ are rank-1, then Eq.~\eqref{eq:app:NRCp_expanded} collapses to the corresponding one for NRC, Eq.~\eqref{eq:G_ave_R_matrix}. Notice that one can evaluate $\overline{G(t)}^{\mathrm{NRC}^+}$ regardless of whether the Hamiltonian spectrum actually satisfies NRC\textsuperscript{+}, and obtain the NRC\textsuperscript{+} estimate mentioned in the main text.

Evidently, one can also express the NRC time average, Eq.~\eqref{eq:G_ave_R_matrix}, in terms of the corresponding projector
\begin{align}
 \overline{G(t)}^{\mathrm{NRC}} = 1 - \frac{1}{d^2} \tr\left[ S_{AA'} \mathcal P_{\mathrm{NRC}} \left( S_{AA'} \right)  \right].
\end{align}
If the Hamiltonian does not satisfy NRC, performing a (possibly nonunique) decomposition $H = \sum_k E_k \ket{\phi_k} \!\bra{\phi_k}$ and evaluating Eq.~\eqref{eq:G_ave_R_matrix} gives rise to the corresponding NRC estimate.

Finally, for the case of Haar random unitaries, one has the corresponding projector $\overline {\mathcal U^{\otimes 2} }^{\mathrm{Haar}}  \coloneqq \mathcal P_{\mathrm{Haar}}$ whose range is given by the algebra generated by $\{ I, S \}$~\cite{goodman2009symmetry}. We evaluate its explicit expression in the next section.

We are now ready to give the proof of \autoref{prop:NRC_upper_bound}.

\begin{proof}
The key observation here is that, by construction,  the range of each projector satisfies
\begin{align}
\Ran \left(  \mathcal P_{H^{(2)}}  \right) \supseteq \Ran \left(  \mathcal P_{\mathrm{NRC}^+} \right) \supseteq \Ran \left(  \mathcal P_{\mathrm{NRC}} \right) \supseteq \Ran \left(  \mathcal P_{\mathrm{Haar}} \right) .
\end{align}
Since all of the above are Hilbert-Schmidt orthogonal projectors, it also follows that
\begin{align}  \label{eq:app:superprojectors_inequality}
 \mathcal P_{H^{(2)}} \ge  \mathcal P_{\mathrm{NRC}^+} \ge  \mathcal P_{\mathrm{NRC}} \ge  \mathcal P_{\mathrm{Haar}} \,.
\end{align}
As a result,
\begin{align}
\braket{S_{AA'}, \mathcal P_{H^{(2)}}(S_{AA'})} \ge \braket{S_{AA'},  \mathcal P_{\mathrm{NRC}^+} (S_{AA'})} \ge \braket{S_{AA'}, \mathcal  \mathcal P_{\mathrm{NRC}}(S_{AA'})} \ge \braket{S_{AA'}, \mathcal P_{\mathrm{Haar}} (S_{AA'})},
\end{align}
from which Eq.~\eqref{eq:comparison_time_avgs_ineq} follows immediately.
\end{proof}

%\subsection*{\autoref{prop:reduced_dynamics}}

%\ReducedDynamics*

%\begin{proof}
%The result follows easily by performing the partial trace over $BB'$ in Eq.~\eqref{eq:G_main}.
%\end{proof}

\subsection*{Proof of Eq.~\eqref{eq:Haar_average}}

The Haar average
\begin{align*}
\overline{G}^{\mathrm{Haar}} = \frac{(d_A^2-1)(d_B^2 - 1)}{d^2 - 1}
\end{align*}
can be derived using fact that $\overline{\mathcal U^{\otimes 2}}^{\mathrm{Haar}}$ is the CPTP orthogonal projector over the algebra generated by $\{ I, S \}$~\cite{goodman2009symmetry}, i.e.,
\begin{align}
\mathcal P_{\mathrm{Haar}} (X) \coloneqq \overline{\mathcal U^{\otimes 2}}^{\mathrm{Haar}} (X) = \frac{1}{2} \sum_{\alpha = \pm 1} \frac{I + \alpha S}{d(d+\alpha)} \braket{I + \alpha S,X} ,
\end{align}  
where $S$ swaps $\mathcal H$ and its duplicate $\mathcal H '$, as usual.
Plugging the above into Eq.~\eqref{eq:G_main}, one gets
\begin{align*}
\overline{G}^{\mathrm{Haar}} = 1 - \frac{1}{2d^2} \sum_{\alpha = \pm 1} \frac{\left|  \braket{I + \alpha S, S_{AA'}} \right|^2}{d(d+\alpha)}
\end{align*}
which, after some simple algebra, simplifies to the announced result.

\subsection*{Theorem~\ref{prop:entropy_production}}

\EntropyProduction*

\begin{proof}
Let us do the $\chi = A$ case. The result relies on the observation that one can express $S_{AA'}$ in Eq.~\eqref{eq:G_reduced_dynamics} through the Haar average~\cite{goodman2009symmetry}
\begin{align}
\int dU  \left( \ket{\psi_U}\!\bra{\psi_U} \right)^{\otimes 2} = \frac{1}{d_A(d_A+1)} \left( I_{AA'} + S_{AA'} \right).
\end{align}
Performing the substitution results in
\begin{align*}
G(t) &= 1  + \frac{1}{d_A^2} \tr \left( S_{AA'}   \right) - \frac{d_A + 1}{d_A} \int dU  \, \tr \left( S_{AA'} \big[ \Uplambda^{\!(A) }_t (\ket{\psi_U}\! \bra{\psi_U}) \big]^{\otimes 2} \right) \\
& = \frac{d_A+1}{d_A} \left( 1 - \int dU \,  \tr \left[   \big( \Uplambda^{\!(A) }_t (\ket{\psi_U}\!\bra{\psi_U}) \big)^2  \right] \right) \\
& = \frac{d_A+1}{d_A}  \int dU \, S_{\mathrm{lin}} \left[ \Uplambda^{\!(A)}_t  (\ket{\psi_U} \! \bra{\psi_U}) \right]
\end{align*}
where we used the fact that $ \Uplambda^{\!(A) }_t (I) = I$ and the identity of Eq.~\eqref{eq:app:swap_trick}.

The $\chi = B$ case follows similarly.
\end{proof}

\subsection*{Proof of Eq.~\eqref{eq:concentration_linear_entropy}}

We need to prove that
\begin{align}
\Prob \left\{ \Big| S_{\mathrm{lin}} \big[ \Uplambda^{\!(\chi) }_t \big(  \ket{\psi}\! \bra{\psi} \big) \big] - \frac{d_\chi}{d_\chi + 1} G(t) \Big| \ge \epsilon \right\} \le \exp \left(  - \frac{d_\chi \epsilon^2}{64} \right)
\end{align}
where $\ket{\psi}$ is a Haar random pure state. We will make use of the concentration of measure machinery, briefly presented before the proof of \autoref{prop:GDoubleConcentration}.

The result follows by the use of Levy's lemma and \autoref{prop:entropy_production}, if one shows that the function 
$f: U(d_\chi) \to \mathbb R$ with $f(V) \coloneqq S_{\mathrm{lin}} \big[ \Uplambda^{\!(\chi)}_t  (\ket{\psi_V} \! \bra{\psi_V}) \big]$ is Lipschitz continuous with $K = 4$. As before, we denote $\ket{\psi_V} \coloneqq V \ket{\psi_0}$ for some (irrelevant) reference state $\ket{\psi_0}$.

Indeed, let us show the Lipschitz continuity. We have
\begin{align*}
\big| f(V) - f(W)  \big| &= \big|  \big\| \Uplambda^{\!(\chi)}_t  (\ket{\psi_V} \! \bra{\psi_V})  \big\|_2^2   -    \big\| \Uplambda^{\!(\chi)}_t  (\ket{\psi_W} \! \bra{\psi_W})  \big\|_2^2  \big|  \\
& =   \Big(  \big\| \Uplambda^{\!(\chi)}_t  (\ket{\psi_V} \! \bra{\psi_V})  \big\|_2   +    \big\| \Uplambda^{\!(\chi)}_t  (\ket{\psi_W} \! \bra{\psi_W})  \big\|_2  \Big) \,   \Big|  \big\| \Uplambda^{\!(\chi)}_t  (\ket{\psi_V} \! \bra{\psi_V})  \big\|_2   -    \big\| \Uplambda^{\!(\chi)}_t  (\ket{\psi_W} \! \bra{\psi_W})  \big\|_2  \Big| \\
& \le 2  \big\| \Uplambda^{\!(\chi)}_t  (\ket{\psi_V} \! \bra{\psi_V})   -     \Uplambda^{\!(\chi)}_t  (\ket{\psi_W} \! \bra{\psi_W})   \big\|_1 \\
& \le 2 \Big\|  \mathcal U _t \Big(  \ket{\psi_V} \! \bra{\psi_V} \otimes \frac{I_{d_{\overline \chi}}}{d_{\overline \chi}} \Big) - \mathcal U _t \Big(  \ket{\psi_W} \! \bra{\psi_W} \otimes \frac{I_{d_{\overline \chi}}}{d_{\overline \chi}} \Big) \Big\|_1 \\ 
& \le 2 \big\| \big(\ket{\psi_V} \! \bra{\psi_V} - \ket{\psi_W} \! \bra{\psi_W}  \big) \otimes  \frac{I_{d_{\overline \chi}}}{d_{\overline \chi}} \big\|_1  = 2 \big\| \ket{\psi_V} \! \bra{\psi_V} - \ket{\psi_W} \! \bra{\psi_W}  \big\|_1  \,,
\end{align*}
where in the second to last line we used the monotonicity of the 1-norm under the partial trace and in the last line that it is unitarily invariant. Utilizing the inequality $\big\| X  \big\|_1 \le \sqrt{\Rank(X)} \left\| X \right\|_2 $, we have
\begin{align*}
\big| f(V) - f(W)  \big| &\le 2 \sqrt{2}  \big\| \ket{\psi_V} \! \bra{\psi_V} - \ket{\psi_W} \! \bra{\psi_W}  \big\|_2 = 4 \sqrt{1 - | \! \braket{\psi_V | \psi_W}  \! |^2 } \\
& \le 4 \sqrt{2 ( 1 - | \! \braket{\psi_V | \psi_W}  \!  |  )} \le  4 \sqrt{2 ( 1 - \Real  \braket{\psi_V | \psi_W}  )}  \\
& \le  4\| \ket{\psi_V} -  \ket{\psi_W}  \| \le 4 \| V - W \|_\infty \\
&\le 4 \| V - W \|_2
\end{align*}
hence one can take $K = 4$.

\subsection*{Proposition~\ref{prop:Choi}}

\Choi*

\begin{proof}
Let us first express the Choi states explicitly as
\begin{align*}
\rho_{\Uplambda^{\!(\chi) }_t} &= \big(  \Uplambda^{\!(\chi) }_t \otimes \mathcal I \big) \ket{\phi^+}\! \bra{\phi^+} = \frac{1}{d_\chi} \sum_{ij} \Uplambda^{\!(\chi) }_t \big(  \ket{i}\! \bra{j} \big) \otimes \ket{i} \!\bra{j} \\
\rho_{\mathcal T^{(\chi)}} & =  \big(  \mathcal T^{(\chi)} \otimes \mathcal I \big) \ket{\phi^+}\! \bra{\phi^+} = \left( \frac{I_\chi}{d_\chi} \right) ^{\!\otimes 2} .
\end{align*}
Writing $S_{\chi\chi'} = \sum_{i,j=1}^{d_\chi} \ket{i}\!\bra{j} \otimes \ket{j}\!\bra{i}$ one also has from Eq.~\eqref{eq:G_reduced_dynamics}
\begin{align*}
G(t) = 1 - \frac{1}{d_\chi^2} \sum_{ij} \big\|  \Uplambda^{\!(\chi) }_t \big( \ket{i} \! \bra{j} \big) \big\|_2^2 \,.
\end{align*}
Thus, expanding the Choi state distance,
\begin{align*}
\big\| \rho_{\Uplambda^{\!(\chi) }_t}  - \rho_{\mathcal T^{(\chi)}} \big\|_2^2 & = \braket{ \rho_{\Uplambda^{\!(\chi) }_t}   - \rho_{\mathcal T^{(\chi)}}, \rho_{\Uplambda^{\!(\chi) }_t}   - \rho_{\mathcal T^{(\chi)}}  } = \braket{ \rho_{\Uplambda^{\!(\chi) }_t}  , \rho_{\Uplambda^{\!(\chi) }_t} } - 2 \braket{ \rho_{\Uplambda^{\!(\chi) }_t}   , \rho_{\mathcal T^{(\chi)}}  } + \braket{  \rho_{\mathcal T^{(\chi)}},  \rho_{\mathcal T^{(\chi)}}  } \\
& = \big\|   \rho_{\Uplambda^{\!(\chi) }_t}  \big\|_2^2 - \frac{1}{d_\chi^2} = \frac{1}{d_\chi^2} \sum_{ij} \big\|  \Uplambda^{\!(\chi) }_t \big( \ket{i} \! \bra{j} \big) \big\|_2^2  -  \frac{1}{d_\chi^2} \\
& = 1 - G(t) -  \frac{1}{d_\chi^2} 
\end{align*}
which is what we wanted.
\end{proof}

\subsection*{Proof of $ \big\| \Uplambda^{\!(\chi)}_t - \mathcal T^{(\chi)} \big\|_{\lozenge} \le d_{\chi}^{3/2} \sqrt{ G_{\max}^{(\chi)}  -  G(t)}$ and an application on information spreading}

We first remind the reader that the diamond norm can be defined as $ \left\| \mathcal X \right\|_{\lozenge} \coloneqq \left\| \mathcal X \otimes \mathcal I_d \right\|_{1,1} $ where $\mathcal I_d$ denotes the identity quantum channel over $\mathcal H \cong \mathbb C^{d}$ and $\left\|  \mathcal X \right\|_{1,1} \coloneqq \sup_{\left\|  A  \right\|_1 = 1} \left\|  \mathcal X(A) \right\|_1$. One of the reasons for this definition is the property that $\left\| \mathcal X \otimes \mathcal Y \right\|_{\lozenge} =  \left\| \mathcal X \right\|_{\lozenge} \left\| \mathcal Y  \right\|_{\lozenge}$, which in general fails for the $\left\| \left( \cdot \right) \right\|_{1,1}$ norm (see, e.g.,~\cite{kitaev2002classical}).

Let us now prove that
\begin{align*}
\sqrt{ G_{\max}^{(\chi)}  -  G(t) } \le \big\| \Uplambda^{\!(\chi)}_t - \mathcal T^{(\chi)} \big\|_{\lozenge} \le d_{\chi}^{3/2} \sqrt{ G_{\max}^{(\chi)}  -  G(t)} \,. \label{eq:ineq_diamond}
\end{align*} 

\begin{proof}
The result follows easily by utilizing the inequalities
\begin{align}
 \big\|\rho_{\mathcal E_1}  - \rho_{\mathcal E_2} \big\|_1 \le \big\|  \mathcal E_1 - \mathcal E_2 \big\|_\lozenge \le d \big\|\rho_{\mathcal E_1}  - \rho_{\mathcal E_2} \big\|_1
\end{align}
that hold for any pair of CPTP maps. The inequality was reported by John Watrous in~\cite{watrous11stackexchange}.
The result follows by use of the inequality $\big\| X \big\|_1 \le \sqrt{d}  \big\| X \big\|_2$ and \autoref{prop:Choi}.
\end{proof}

As an additional application of Eq.~\eqref{eq:ineq_diamond}, we can utilize it to bound from above the fraction of time such that $\big\| \Uplambda^{\!(\chi)}_t - \mathcal T^{(\chi)} \big\|_{\lozenge}  \ge \epsilon $ holds true. This can be done by combining Eq.~\eqref{eq:ineq_diamond} with our earlier time averages. The result
\begin{align} \label{eq:Markov}
\Prob \big\{ t \; \big|  \; \big\| \Uplambda^{\!(\chi)}_t - \mathcal T^{(\chi)} \big\|_{\lozenge}  \ge \epsilon   \big\} \le \frac{2  d_{\chi}^{3/2}}{\epsilon  d_{\overline \chi} } \kappa \,,
\end{align} 
where $\kappa \coloneqq \sqrt{1 + \dfrac{d_{\overline \chi}^2}{2} \big( \overline{G}^{\mathrm{Haar}} - \overline{G(t)} \big)}$, demonstrates in yet another way that if $ d_{\overline \chi} \gg d_{\chi} $ and $\kappa = O(1)$ (i.e., the equilibration is sufficiently close to the Haar estimate), then the reduced evolution is necessarily close to the maximally mixing one for a large fraction of time.

\begin{proof}
Our starting point will be inequality~\eqref{eq:ineq_diamond}, $\big\| \Uplambda^{\!(\chi)}_t - \mathcal T^{(\chi)} \big\|_{\lozenge} \le d_{\chi}^{3/2} \sqrt{ G_{\max}^{(\chi)}  -  G(t)}\,$. By taking the time average of both sides, and then using the concavity of the square root, we obtain
\begin{align*}
\overline{\big\| \Uplambda^{\!(\chi)}_t - \mathcal T^{(\chi)} \big\|_{\lozenge}} \le d_{\chi}^{3/2} \sqrt{  G_{\max}^{(\chi)} - \overline {G(t)}} \le  d_{\chi}^{3/2} \sqrt{ \big( G_{\max}^{(\chi)} - \overline G ^{\mathrm{Haar}} \big)  + \big( \overline G ^{\mathrm{Haar}} - \overline {G(t)} \big) } \le 2 \frac{d_\chi^{3/2}}{d_{\overline \chi}} \kappa  \,,
\end{align*}
%Our next step is to trade the exact time-average for the Haar one from Eq.~\eqref{eq:Haar_average}, which is always valid by the inequality \eqref{eq:comparison_time_avgs_ineq}. This results to
%\begin{align*}
%\overline{\big\| \Uplambda^{\!(\chi)}_t - \mathcal T^{(\chi)} \big\|_{\lozenge}} \le  2 \frac{d_\chi^{3/2}}{d_{\overline \chi}} 
%\end{align*}
where we approximated the difference
\begin{align*}
 G_{\max}^{(\chi)} - \overline {G(t)}^{\mathrm{Haar}}  = \frac{(d_\chi^2 - 1)^2}{d_\chi^2 (d^2-1)} \le \frac{2}{d_{\overline \chi}^2} \,.
\end{align*}
Finally, Eq.~\eqref{eq:Markov} follows by the use of Markov's inequality.
%As an example, one can choose $\epsilon = d_{\overline \chi}^{-2/3}$ which gives 
%\begin{align}
%\Prob \big\{ t \; \big|  \; \big\| \Uplambda^{\!(\chi)}_t - \mathcal T^{(\chi)} \big\|_{\lozenge}  \ge d_{\overline \chi}^{-2/3}   \big\} \le \frac{2  d_{\chi}^{3/2}}{ d_{\overline \chi}^{1/3} } 
%\end{align}
%demonstrating that, for a large environment $d_{\overline \chi} \gg  d_\chi$, the reduced evolution is forced to be close to the maximally mixing one for a large fraction of time.
\end{proof}

\section{Haar measure, unitary $k$-designs and the bipartite OTOC} \label{sec:app:measures}

%\subsection*{Alternatives to the Haar measure}

Here we discuss in more details how the Haar measure in the definition of the bipartite OTOC, Eq.~\eqref{eq:defnition_G}, can be replaced by other possible averaging choices, in a way that Eq.~\eqref{eq:G_main} (and everything that stems from it) remains valid.

Let us first recall the definition of a (unitary) $k$-design~\cite{divincenzo2002quantum,renes2004symmetric,
scott2006tight,gross2007evenly,roberts2017chaos}. Consider an ensemble of unitary operators $\Lambda = \{ ( p_i, U_i ) \}_i$ and define the family of CPTP maps
\begin{align}
\mathcal E ^{(k)} _{\Lambda}  &\coloneqq  \sum_i p_i U_i^{\otimes k} (\cdot) U_i^{\dagger \otimes k} \label{eq:app:Haar_channel}  \\
\mathcal E ^{(k)} _{\mathrm{Haar}}  & \coloneqq \int dU \, U^{\otimes k} (\cdot) U^{\dagger \otimes k}  \label{eq:app:ensemble_channel}
\end{align}
for $ k \in \mathbb N $. The ensemble $\Lambda $ forms a $k$-design if $\mathcal E ^{(k)} _{\Lambda} = \mathcal E ^{(k)} _{\mathrm{Haar}}$. In words, a $k$-design emulates Haar averaging up to (at least) the $k\textsuperscript{th}$ moment.

Now, let us investigate what is the freedom over the possible probability measures of $V_A$ and $W_B$ in Eq.~\eqref{eq:defnition_G}, such that Eq.~\eqref{eq:G_main} holds true without modification. 
It is easy to see, by the proof of \autoref{prop:GAverageUnitaries}, that we are in fact looking for a unitary ensemble $\Lambda$ retaining the validity of Eq.~\eqref{eq:app:u_avg_vectorized}. In turn, the latter is just a vectorized form of the $1$-design condition $\mathcal E ^{(1)} _{\Lambda} = \mathcal E ^{(1)} _{\mathrm{Haar}}$. One can therefore substitute the Haar measure over $U(d_A)$ and $U(d_B)$ with $1$-designs over the corresponding spaces; the full Haar randomness is not probed by the OTOC~\cite{roberts2017chaos}.

Moreover, $1$-designs factorize, i.e., if $\Lambda_1 =  \{ ( p^{(1)}_i, U^{(1)}_i ) \}_i$ and $\Lambda_2 =  \{ ( p^{(2)}_j, U^{(2)}_j ) \}_j$ are 1-designs over $\mathcal H_A$ and $\mathcal H_B$ respectively, then $\Lambda_1 \otimes \Lambda_2 \coloneqq \{ ( p^{(1)}_i p^{(2)}_j, U^{(1)}_i \otimes U^{(2)}_j ) \}_{ij}$  is a 1-design over $\mathcal H = \mathcal H_A \otimes \mathcal H_B$. This follows just by the 1-design condition in the form of Eq.~\eqref{eq:app:u_avg_vectorized} and the fact that the swap operator over the duplicated space $\mathcal H \otimes \mathcal H'$ factorizes $S_{AB;A'B'} = S_{AA'}S_{BB'}$.

This last fact has an important implication for the physically relevant case of many-body systems. Consider the case where $\mathcal H_\chi = 
\bigotimes_i \mathcal H_\chi^{(i)}$ for $\chi = A,B$, i.e., when $A$ and $B$ are made up of (not necessarily identical) individual subsystems. Then the OTOC of Eq.~\eqref{eq:defnition_G} remains unchanged if the averages $\int dV_A$ and $\int dW_b$ are replaced by the unitary ensemble $\bigotimes_i \Lambda_\chi^{(i)}$, where each $\Lambda_\chi^{(i)}$ is a $1$-design on $\mathcal H_\chi^{(i)}$. In other words, it is always enough to average over unitary operators that factorize completely. For instance, in the case of a spin-$1/2$ many-body system $\mathcal H_\chi^{(i)} \cong \mathbb C^2$ such an example is given by the Pauli $1$-design $\Lambda^{(i)}_{\chi,\mathrm{Pauli}} \coloneqq \{ 1/4 , \sigma_k \}_{k=0}^3$~\cite{webb2015clifford}.

\section{Estimating the bipartite OTOC via linear entropy measurements of random pure states}  \label{sec:app:linear_entropy}

Here we present a basic protocol, stemming directly from \autoref{prop:entropy_production}, for the estimation of the bipartite OTOC via repeated measurements of a single expectation value.

\begin{figure}[h]
\centering 
\includegraphics[width=.5\textwidth]{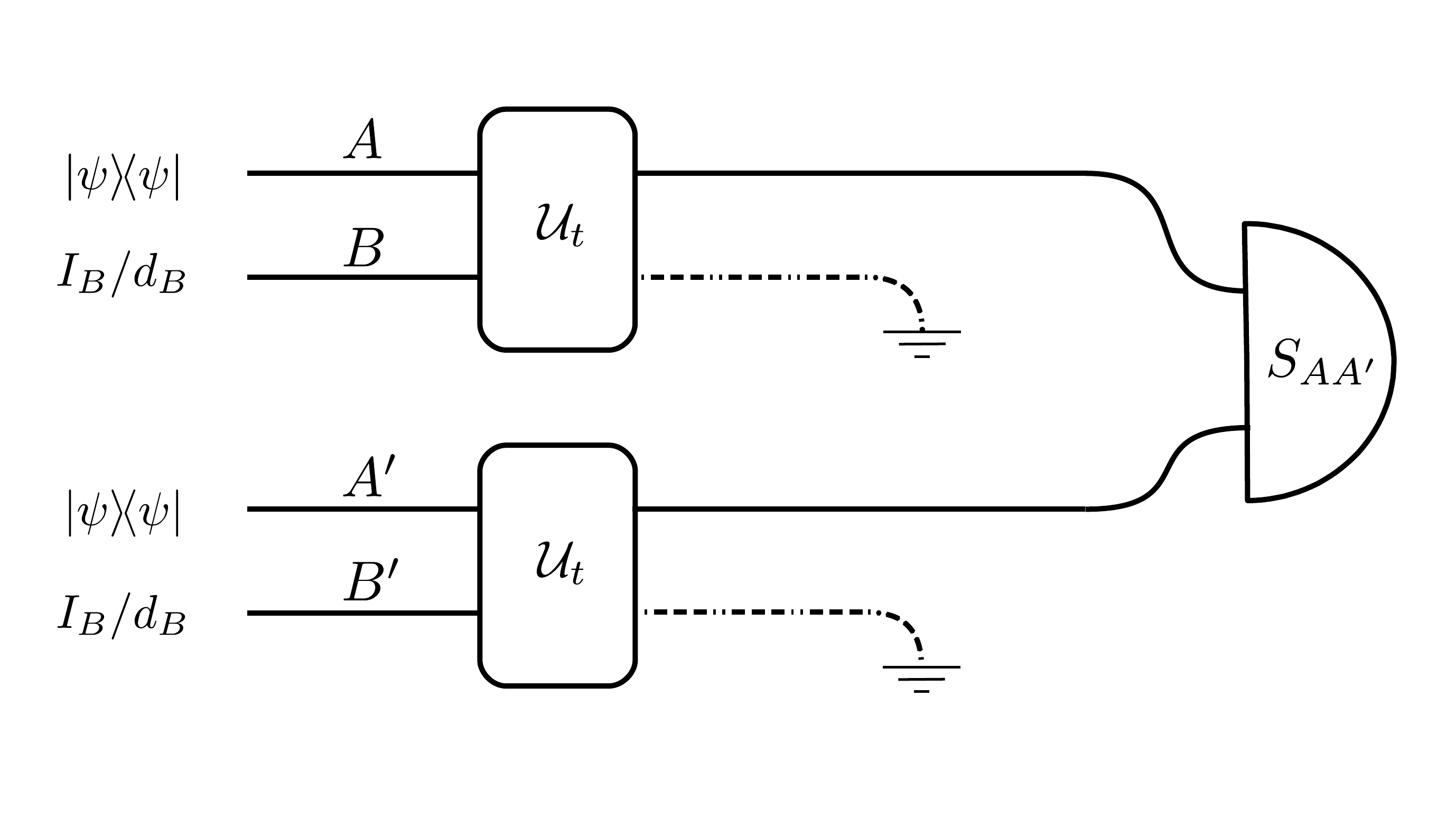}
\caption{Protocol to for the estimation of the purity $1 - S_{\mathrm{lin}} \left[ \Uplambda^{(A)}_t (\ket{\psi}\!\bra{\psi}) \right]$ according to Eq.~\eqref{eq:entropy_production}. The resulting purity constitutes also an estimate of the bipartite OTOC, up to a simple proportionality factor. The final measurement of the swap operator can be realized, for instance, by measuring the expectation value of $A$ and $A'$ over any preferred product basis $\{ \ket{i} \otimes \ket{j} \}_{i,j = 1}^{d_A}$, without the need for coherences.}.
\label{fig:circuit}
\end{figure}

As pointed out in the main text, the linear entropy of a state can be expressed as an expectation value, $1 - S_{\mathrm{lin}} (\rho) = \tr \left( S \rho^{\otimes 2} \right)$ at the expense of requiring two copies of the state $\rho$, though uncorrelated. Combining \autoref{prop:entropy_production} with the above observation, one can realize a simple protocol for estimating the bipartite OTOC via measuring the expectation value of the swap operator over pairs of randomly generated states $\ket{\psi} \in \mathcal H_A$. We schematically draw the protocol in \autoref{fig:circuit}.

Averaging the resulting expectation value over Haar random pure states $\ket{\psi}$ converges to the exact value of the bipartite OTOC. In light of Eq.~\eqref{eq:concentration_linear_entropy}, the expected number of sample for this convergence to a given accuracy drops fast as $d_A$ increases. Clearly, the corresponding protocol with the roles of $A$ and $B$ interchanged is formally equivalent.

Along conceptually similar lines, there have been a number of proposals for probing the linear entropy of a state in an experimentally accessible way. For example, in a recent experiment~\cite{islam2015measuring} quantum purity (which is directly related to the second-order R\'enyi entanglement entropy) was measured by interfering two uncorrelated but identical copies of a many-body quantum state; similar ideas have also been considered previously~\cite{daley2012measuring,ekert2002direct,PhysRevLett.93.110501,bovino2005direct}. In particular, this scheme neither requires full quantum state tomography nor the use of entanglement witnesses to estimate entanglement of a quantum state. %Moreover, the protocol employed is based on expressing the linear entropy as an expectation value, as discussed above, and this allows us to measure \textit{nonlinear} functions of the density matrix by measuring a single operator (at the expense of creating multiple, identical copies). 

Furthermore, there have been recent proposals for protocols based on measurements over random local bases that can probe entanglement given just a single copy of the quantum state, and, in this sense, go beyond traditional quantum state tomography. The main idea consists of directly expressing the linear entropy~\cite{brydges2019probing,elben2019statistical}, as well as other functions of the state~\cite{huang2020predicting}, as an ensemble average of measurements over random bases. Related ideas have also been adapted to probe OTOCs~\cite{vermersch2019probing,joshi2020quantum} and mixed state entanglement~\cite{elben2020mixed}.

\end{document}